\documentclass[11pt,english]{article}
\pdfoutput=1
\usepackage[latin9]{inputenc}
\usepackage[letterpaper]{geometry}
\usepackage{amsmath,amssymb,amsthm}
\usepackage{graphicx}
\usepackage{setspace}
\usepackage{babel}

\theoremstyle{plain}
\newtheorem{theorem}{Theorem}
\newtheorem*{theorem*}{Theorem}

\newtheorem*{corollary*}{Corollary}

\def\g{\gamma}
\def\O{\Omega}
\def\s{\sigma}
\def\l{\lambda}
\def\t{\theta}
\def\F{\mathcal{F}}
\def\pfi{\varphi}

\def\ds{\displaystyle}

\makeatletter

\date{}

\makeatother

\begin{document}

\title{Synchronization of Huygens' clocks and the Poincar\'e method}

\author{Vojin Jovanovic\\
 Systems, Implementation \& Integration\\
 Smith Bits, A Schlumberger Co.\\
 1310 Rankin Road\\
 Houston, TX 77073\\
 e-mail: fractal97@gmail.com \and Sergiy Koshkin\\
 Computer and Mathematical Sciences\\
 University of Houston-Downtown\\
 One Main Street, \#S705\\
 Houston, TX 77002\\
 e-mail: koshkins@uhd.edu}
\maketitle
\begin{abstract}
We study two models of connected pendulum clocks synchronizing their oscillations, a phenomenon originally observed by Huygens. The oscillation angles are assumed to be small so that the pendulums are modeled by harmonic oscillators, clock escapements are modeled by the van der Pol terms. The mass ratio of the pendulum bobs to their casings is taken as a small parameter. Analytic conditions for existence and stability of synchronization regimes, and analytic expressions for their stable amplitudes and period corrections are derived using the Poincar\'e theorem on existence of periodic solutions in autonomous quasi-linear systems. The anti-phase regime always exists and is stable under variation of the system parameters. The in-phase regime may exist and be stable, exist and be unstable, or not exist at all depending on parameter values. As the damping in the frame connecting the clocks is increased the in-phase stable amplitude and period are decreasing until the regime first destabilizes and then disappears. The results are most complete for the traditional three degrees of freedom model, where the clock casings and the frame are consolidated into a single mass.

\end{abstract}

\newpage

\section{Introduction}

In 1665 Huygens reported that two pendulum clocks attached to a common frame tended, after a while, to run in unison with their pendulums always pointing in mutually opposing directions. He eventually attributed this "sympathy" to imperceptible motions of the frame that transfer motion from one clock to the other. With some 
modifications, Huygens' experiments were repeated by Ellicott, Euler and Poisson, among others, who observed similar behavior, see descriptions in \cite{Kort}. Original attempts to give a quantitative description of this sympathy were based on linear theory. Korteweg in a paper of 1905 \cite{Kort} showed that the motion of the linearized system decomposes into three normal modes, two of which involve sizable oscillations of the frame. He then argued heuristically that damping in the frame siphons energy away from these two modes, eventually leaving only the mode responsible for the anti-phase synchrony. While an improvement over Huygens' intuitive explanation, Korteweg's theory left all non-linear effects of clock escapements and coupling out of his quantitative analysis. Nonetheless, upgraded variations of modal analysis appear in some recent approaches to the problem \cite{Ben, Sen}.

By the end of the 19th century coordinated behavior in systems of interacting self-excited oscillators was observed in many other acoustic, electric and mechanical systems, the phenomenon later came to be known as
synchronization. About the same time, motivated by applications to celestial mechanics, Poincar\'e adapted the method of small parameter to finding periodic solutions to non-linear systems. Somewhat later, the Poincar\'e method was applied to 
synchronization problems by Blekhman and others, see \cite{Bl2} and references therein, providing a first non-linear theory of the effect. Since the method requires equations of motion to be analytic in phase variables Blekhman used the simplest analytic model of a self-excited oscillator, the van der Pol model, to describe clock escapements. His analysis predicted not only the previously observed anti-phase synchronization, but also the in-phase synchronization, when the pendulums always point in the same direction. According to Blekhman, both regimes always co-exist in the same system being triggered by different initial conditions, he also reported observing both in experiments. 

In the last decade interest in Huygens' sympathy was renewed by the work of Bennett et al. \cite{Ben}, who took great care to reproduce Huygens' experiments and developed a new non-linear model to describe them. Their work builds on Korteweg's modal analysis, but equations of motion include terms modeling clock escapements with impulsive kicks. Unlike Blekhman's their model does not predict a stable in-phase synchronization, but it does predict asymptotic vanishing of oscillations in one of the clocks, a phenomenon they call beating death. This is a direct consequence of having a threshold angle for escapements to engage, if at some point the amplitude of oscillations falls below it energy loss forces oscillations to damp out. The authors even expressed puzzlement over Blekhman's results since neither they nor Huygens ever observed the in-phase synchronization. However, Czolczynski et al. \cite{Cz2} later confirmed co-existence of both regimes experimentally using modern clocks with non-impulsive escapements. Numerical simulations in 
\cite{Dil,Frad} also show co-existence for some parameter values. The common trait in all three cases is that the escapements involved do not have an engagement threshold.

We should also mention that the in-phase synchronization was observed by Pantaleone \cite{Pan} and others 
\cite{Cz1,Oud} after replacing pendulums with metronomes. Moreover, for small damping in the frame the 
in-phase synchronization was the only one observed, in contrast to Huygens' setup. The reason is that metronomes swing with amplitudes too large to be approximated by harmonic oscillators, changing the situation entirely. Existence of the  the in-phase synchronization for large amplitudes was proved in \cite{Kuz} using the van der Pol model for escapements.
 
Our main purpose in this paper is to apply the Poincar\'e method to the Huygens system and to compare predictions it makes to known observations. Our application of the Poincar\'e method builds on the original work of Blekhman 
\cite{Bl2}. As already mentioned, one obvious limitation of the method is that all terms in the equations have to be analytic, so realistic escapements with impulsive kicks can not be accomodated. However, van der Pol approximation appears to be rather accurate as long as the pendulum amplitudes stay above the engagement threshold, see \cite{Pan}. On the other hand, the method offers several important benefits. First, there is a natural small parameter in the system, the mass ratio of pendulum bobs to the frame, which incorporates heavy clock casings. Second, one can take into account not only escapements bit also the non-linear coupling term, which remains after pendulum equations are partially linearized for small angles. In \cite{Ben} and \cite{Dil} this term had to be dropped because both approaches rely on the system's evolution being linear between escapement kicks. However, its magnitude is comparable to the linear terms kept. Finally, the Poincar\'e method identifies synchronization regimes rigorously without heuristic pre-simplifications required for modal analysis, and gives analytic expressions for their stability ranges and amplitudes.

One of the method's key predictions, namely co-existence of the in-phase and the anti-phase regimes, appears to be at variance with some observations. We will show in Sections \ref{S2}, \ref{S3} that this prediction is not due to the method's limitations, but arises from treating damping in the frame as a small parameter. In Section \ref{S2} we assume both coupling of the pendulums to the frame and damping in the frame to be small. To avoid neglecting the second derivative in the frame's position, we transform the system into a first order form in a way different from that in 
\cite{Bl2}. However, our conclusions are the same in this setting: the in-phase and the anti-phase regimes co-exist for the same values of system parameters. 

On the other hand, once damping in the frame is treated as a regular parameter in Section \ref{S3} one sees that the two regimes are not created equal. While the anti-phase synchronization remains robustly stable under variation of its values, the in-phase regime is only stable when the damping is very small relative to coupling. As the damping gets larger, it destabilizes and then disappears altogether. Although the Poincare theorem, as any asymptotic result, does not prescribe what "small" is, simulations in Section \ref{S5} suggest that in Huygens' experiments values of damping were outside of the co-existence range. In other words, damping in the frame is a bifurcation parameter and experimental values for it can be both below or above the bifurcation values. To illustrate the point, numerical simulations in Section \ref{S5} are compared to experimental results in \cite{Ben,Cz2}.

As in \cite{Ben}, additional clues are provided by a preliminary linear analysis in Section \ref{S1} on the system obtained by disregarding all non-linear terms. Although synchronization is often believed to be a non-linear phenomenon, it turns out not to be entirely so. Even the linear system has the sum of pendulum angles converge to zero, which is an anti-phase synchronization of sorts. Of course, due to linearity stable amplitudes depend on initial conditions, it is their convergence to invariant values that is a non-linear effect. Also, in the linear model pendulums started with exactly equal phases have their oscillations asymptotically damp out. Although van der Pol terms alter this asymptotic outcome, numerical simulations show that the system may initially follow the linear pattern.

In addition to the traditional three degrees of freedom model, where the clock casings and the connecting frame are modeled by a single mass, we also consider a four degrees of freedom system, where each casing is represented by a separate mass (Section \ref{S4}). In the linear approximation this system reduces to just three degrees of freedom because the angle difference between the pendulums decouples from the rest of the variables. Similar system but with a piecewise-linear escapement term was considered in \cite{Dil}, where existence of a stable anti-phase synchronization is proved for some parameter range.

 \section{Equations of motion and linear synchronization}\label{S1}

We start with a general setup of $n$ identical pendulums ('clocks')
of mass $m$ and length $l$ attached to a rigid frame of mass $M$,
which is restricted to move horizontally. In its turn, the frame is attached to a fixed support by 
a weightless spring with stiffness $k$ and damping $c$, Figure \ref{fig:OneMassSys}.
Pendulum angles with the vertical are denoted by $\t_{i}$ and position
of the frame is given by $x$. Horizontal restriction reflects the original setup of Huygens, where clocks were hanging off a wooden beam resting on the backs of two chairs. Experiments show that behavior of the system changes substantially if non-horizontal motion is allowed, and synchronization may not occur at all \cite{Cz2}.
\begin{figure}[htbp]
\centering
\includegraphics[width=3.2in]{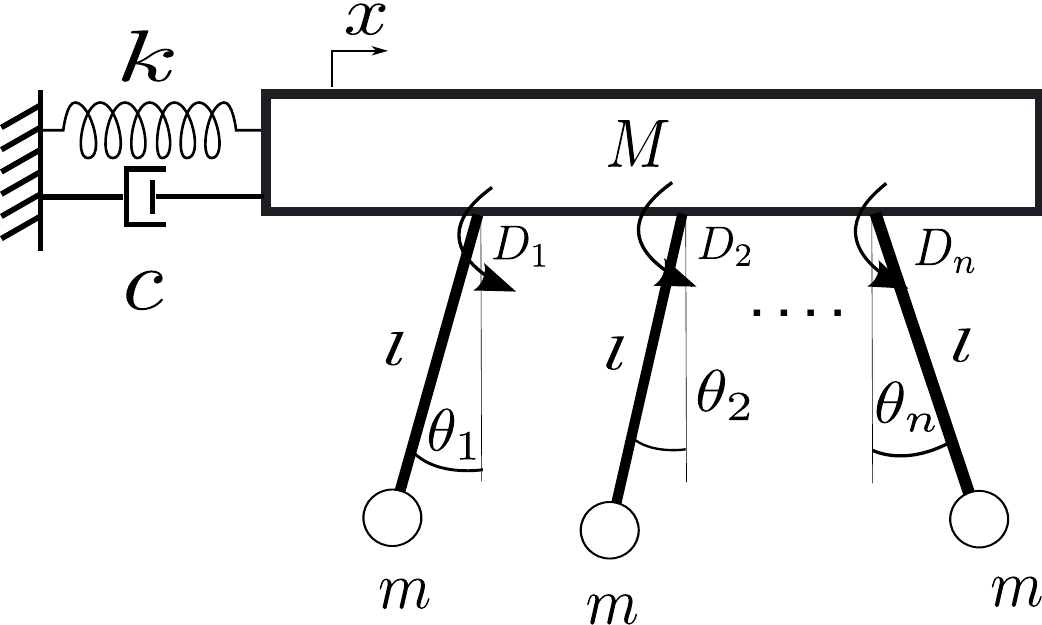}
\caption{\label{fig:OneMassSys}Rigid frame with pendulums attached to a fixed support.} 
\end{figure}
Equations of motion can be derived using the Lagrangian approach, see e.g. \cite{Cz1}, 
\begin{alignat*}{1}
\begin{cases}
ml^{2}\ddot{\theta}_{i}+mgl\sin(\theta_{i})+ml\ddot{x}\cos(\theta_{i})=D_{i},\hspace{3em} i=1,...,n\\
(M+nm)\ddot{x}+c\dot{x}+kx+ml\sum_{i=1}^{n}\ddot{\theta}_{i}\cos(\theta_{i})=ml\sum_{i=1}^{n}\dot{\theta}_{i}^{2}\sin(\theta_{i})
\end{cases}
\end{alignat*}
Here $D_{i}:=D(\t_{i},\dot{\t}_{i})$ denote driving torques provided
by the clock escapement mechanisms. The Poincar\'e method, which is
our method of choice, requires all terms to be analytic in their variables.
Following Blekhman \cite{Bl2}, we choose the simplest analytic term that admits a limit cycle,
the van der Pol term $D(\t,\dot{\t}):=e(\g^{2}-\t^{2})\dot{\t}$. Here $e$ measures 
the strength of the escapement and $\g$ is the critical angle, at 
which the escapement switches from boosting to damping.
This model does not always work at small amplitudes due to engagement threshold
and impulsive kicks in some escapements, but on average gives a good
approximation when the pendulums are close to their limit cycles \cite{Pan}.
We also assume that oscillation angles remain sufficiently small
so that the pendulums can be approximated by harmonic oscillators.
Taking $\sin\t_{i}\approx\t_{i}$ and $\cos\t_{i}\approx1$ we get
a partially linearized system 
\begin{alignat*}{1}
\begin{cases}
ml^{2}\ddot{\theta}_{i}+mgl\theta_{i}+ml\ddot{x}=D_{i},\hspace{3em} i=1,...,n\\
(M+nm)\ddot{x}+c\dot{x}+kx+ml\sum_{i=1}^{n}\ddot{\theta}_{i}=ml\sum_{i=1}^{n}\dot{\theta}_{i}^{2}\theta_{i}\,.
\end{cases}
\end{alignat*}
Rescaling the time as $\ds{\tau=\sqrt{\frac{g}{l}}\, t}$ and introducing
a dimensionless frame position variable $\ds{y=\frac{x}{l}}$ we get
a dimensionless system 
\begin{equation}\label{DimssSys}
\begin{cases}
\ddot{\theta}_{i}+\theta_{i}+\ddot{y}=F_{i},\hspace{3em} i=1,...,n \\
\ddot{y}+\sigma\dot{y}+\Omega^{2}y+\beta\sum_{i=1}^{n}\ddot{\theta}_{i}
=\beta\sum_{i=1}^{n}\dot{\theta}_{i}^{2}\theta_{i}\,,
\end{cases}
\end{equation}
where $\ds{F_{i}=\frac{D_{i}}{mgl}}$, $\ds{\s=\frac{c}{(M+nm)\sqrt{g/l}}}$,
$\ds{\O^{2}=\frac{kl}{(M+nm)g}}$ and $\ds{\beta=\frac{m}{M+nm}}$.
Note that the only remaining non-linearities are on the right hand
side of \eqref{DimssSys}, namely the van der Pol terms and the non-linear
coupling term $\ds{\beta\sum_{i=1}^{n}\dot{\t}_{i}^{2}\t_{i}}$. Before
proceeding with analysis of this system it will be instructive to
investigate a linear version obtained by simply dropping the right hand
sides in \eqref{DimssSys}. The resulting linear system 
\begin{equation}\label{LinSys}
\begin{cases}
\ddot{\theta}_{i}+\theta_{i}+\ddot{y}=0,\hspace{3em} i=1,...,n \\
\ddot{y}+\sigma\dot{y}+\Omega^{2}y+\beta\sum_{i=1}^{n}\ddot{\theta}_{i}=0
\end{cases}
\end{equation}
is used in \cite{Ben} to attempt a simple explanation of Huygens' observations.
Asymptotic behavior of this system can be analyzed using a Lyapunov
function. It turns out that the sum of pendulum angles asymptotically
vanishes. 
\begin{theorem}\label{Lyap} Suppose that all coefficients in \eqref{LinSys}
are non-negative and $\g>0$. Let $\t_{i}=\t_{i}(t)$ be solutions
to it with arbitrary initial conditions, then $\ds{\t_{1}+\dots+\t_{n}\xrightarrow[t\to\infty]{}0}$.
\end{theorem} 
\begin{proof} Set $\ds{\t:=\t_{1}+\dots+\t_{n}}$, then adding the first $n$ equations in \eqref{LinSys} we get 
\begin{equation}\label{SumSys}
\begin{cases}
\ddot{\theta}+\theta+n\ddot{y}=0\\
\ddot{y}+\sigma\dot{y}+\Omega^{2}y+\beta\ddot{\theta}=0\,.
\end{cases}
\end{equation}
Define the Lyapunov function by 
\begin{gather*}
E:=\beta(\t^{2}+\dot{\t}^{2})+n(\O^{2}y^{2}+\dot{y}^{2})+2\beta n\dot{\t}\dot{y}=\beta\t^{2}+(\beta-n\beta^{2})\dot{\t}^{2}+n\O^{2}y^{2}+n(\beta\dot{\t}+\dot{y})^{2}.
\end{gather*}
Then $E\geq0$ as long as $\beta-n\beta^{2}\geq0$ or $0\leq\beta\leq\frac{1}{n}$.
Given the definition of $\beta$ above this holds even with strict
inequalities. Moreover, if $0<\beta<\frac{1}{n}$ then $E=0$ if and
only if $\dot{\t}=\t=\dot{y}=y=0$, i.e. $E$ is positive definite.
On the other hand, by \eqref{SumSys} 
\begin{gather*}
\dot{E}:=2\beta\dot{\t}(\ddot{\t}+\t)+2n\dot{y}(\ddot{y}+\O^{2}y)+2\beta n(\ddot{\t}\dot{y}+\dot{\t}\ddot{y})=-2\beta n\dot{\t}\ddot{y}-2n(\beta\ddot{\t}\dot{y}+2\g\dot{y}^{2})+2\beta n(\ddot{\t}\dot{y}+\dot{\t}\ddot{y})=-4n\g\dot{y}^{2}.
\end{gather*}
By the LaSalle invariance principle any trajectory of \eqref{SumSys}
asymptotically converges to a trajectory contained in $\dot{E}^{-1}(0)$
since $\dot{E}\leq0$. But for $\g>0$ any latter trajectory must
have $\dot{y}=0$, hence $\ddot{y}=0$ and $y=\text{const}$.
The system then implies $\ddot{\t}=-\t=\text{const}$ and $\dot{\t}=\t=0$,
so any trajectory within $\dot{E}^{-1}(0)$ does not leave the origin. 
Therefore, $\ds{y\xrightarrow[t\to\infty]{}0}$ and $\ds{\t=\t_{1}+\dots+\t_{n}\xrightarrow[t\to\infty]{}0}$.
\end{proof} 
Now let us focus on the case of two pendulums that we are most interested in.
\begin{corollary*} Let $n=2$ in Theorem \ref{Lyap} and set $\delta(t):=\t_{1}(t)-\t_{2}(t)$. 
Then asymptotically\\ 
$\begin{cases}\t_{1}(t)\sim\frac12A\,\sin(t+\phi)\\\t_{2}(t)\sim\frac12A\,\sin(t+\phi+\pi)\end{cases}$
with $A=\sqrt{\delta(0)^2+\dot{\delta}(0)^2}$ and $\sin\phi=\delta(0)/A$.
\end{corollary*}
\begin{proof}
Subtracting the equation for $\t_{2}$ in \eqref{LinSys} from the one for $\t_{1}$ we have $\ddot{\delta}+\delta=0$. Therefore, $\delta(t)=A\sin(t+\phi)$ with $A$ and $\phi$ as in the statement of the corollary. Since asymptotically
$\t_{1}(t)\sim-\t_{2}(t)$ by Theorem \ref{Lyap} it follows that $\t_{1}(t)\sim\frac12\delta(t)\sim-\t_{2}(t)$.
\end{proof} 
\begin{figure}[htbp]
\begin{centering}
(a)\includegraphics[scale=0.35]{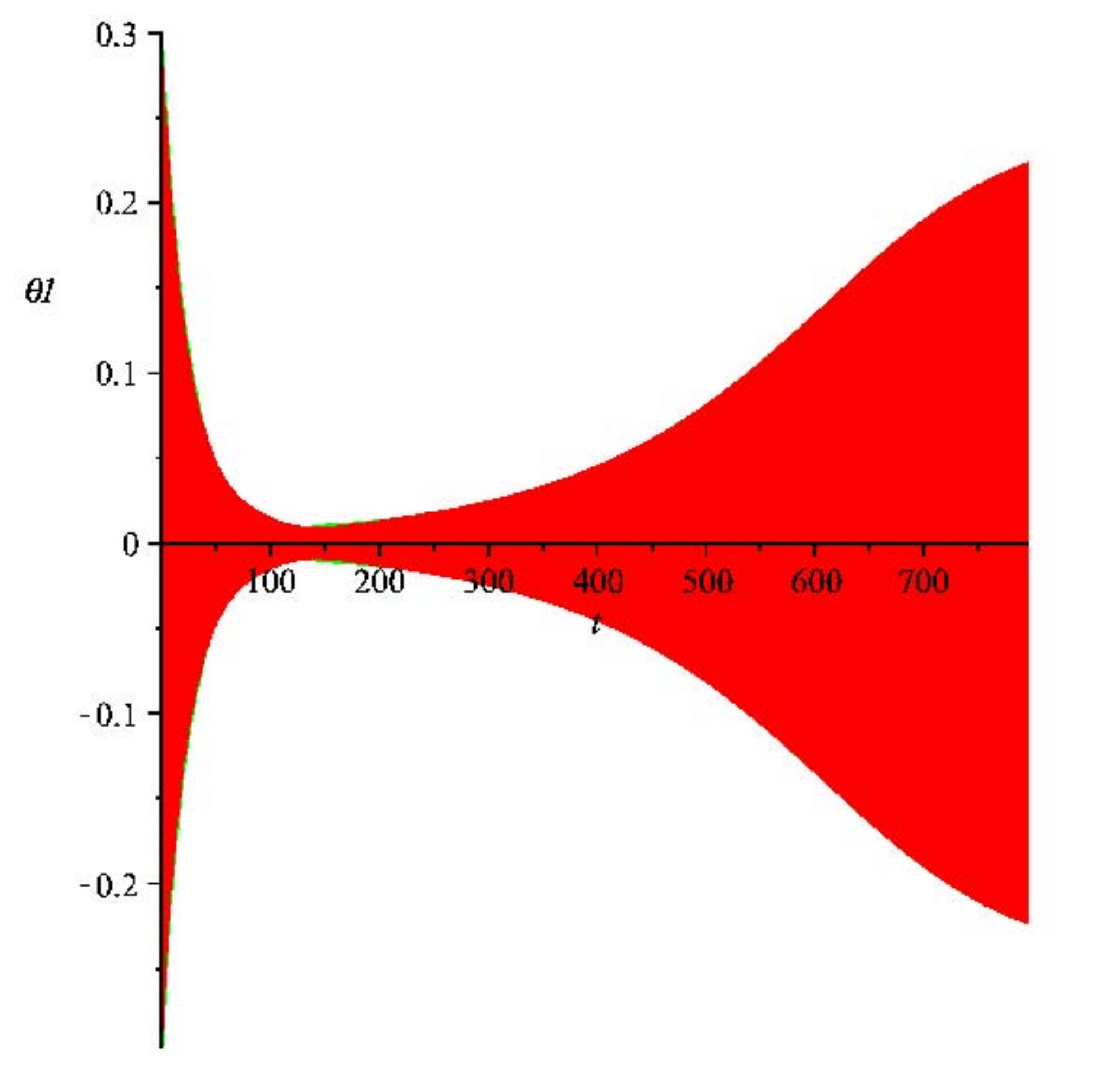}\quad(b)\includegraphics[scale=0.35]{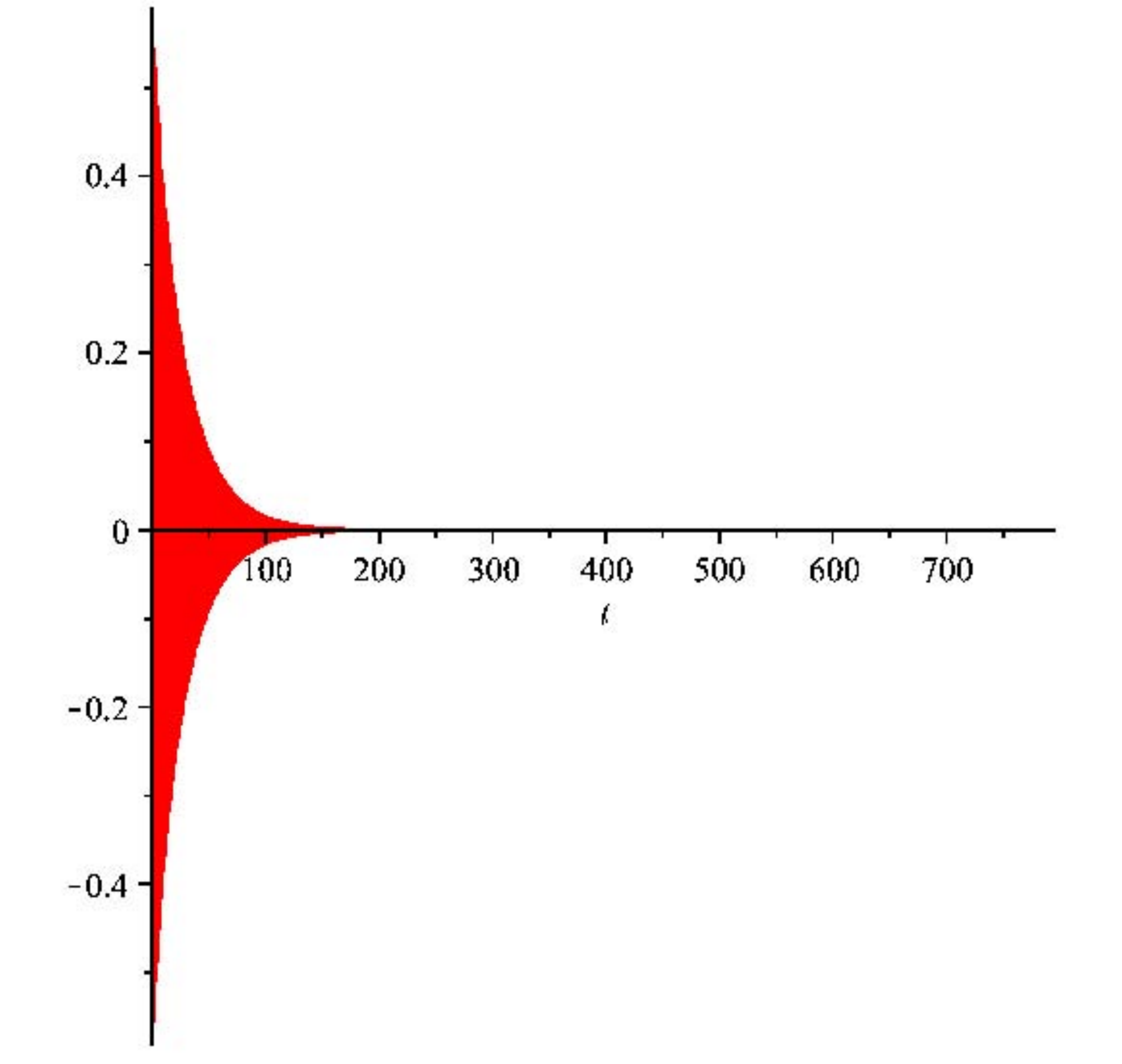}
\caption{\label{Bnwhale}Simulation of \eqref{DimssSys} with $\s=1.16$, $\beta=0.013$, $\O^2=0.0014$, $\g=0.122$ and 
$e/mgl=0.134$. Time is measured in the number of cycles of a decoupled harmonic oscillator with the same parameters.
(a) $\theta_1(t)$ and $\theta_2(t)$ (indistinguishable at this resolution) superimposed, (b) $\theta_1(t)+\theta_2(t)$. Linear pattern is followed for about $150$ cycles.}
\end{centering}
\end{figure}
Thus, pendulum angles asymptotically align in anti-phase for any initial values. The difference with true synchronization is that the stable amplitude $A$ depends on initial values. In particular, if the pendulums are started at exactly in-phase, $\delta(0)=\dot{\delta}(0)=0$, then $A=0$ and they will both asymptotically damp out. Even if initial phases are not exactly the same the asymptotic amplitude can be very small. Numerical simulations show that for large $\s$ and small coupling $\beta$ the non-linear system may follow this linear pattern for many cycles before the van der Pol terms break the trend, see Figure \ref{Bnwhale}. If an impulsive escapement is used instead, and the engagement threshold for it is above this asymptotic amplitude, one can expect that one or both pendulums will stop beating in finite time.

\section{Small damping in the frame}\label{S2}

Poincar\'e method relies on non-linear terms entering equations being multiplied
by a small parameter. In practice, this can be interpreted as some dimensionless
parameter of the system being 'sufficiently' small. In our case, it is natural to assume that $\ds{\beta=\frac{m}{M+nm}}$
is small since masses of the pendulum bobs $m$ are much smaller
than the consolidated mass $M$ of the clock casings and the frame.
We will also assume that the escapement strength $\ds{\varepsilon:=\frac{e}{mgl}}$
of the van der Pol term $F_{i}=\varepsilon(\g^{2}-\t_{i}^{2})\dot{\t}_{i}$
is small for the purposes of the Poincare theorem (which may in fact be large in practical terms). 
Eliminating $\ddot{y}$ from the pendulum equations in \eqref{DimssSys}
we get 
\begin{equation}\label{SecDerSys}
\begin{cases}
\ddot{\theta}_{i}+\theta_{i}-\sigma\dot{y}-\Omega^{2}y=F_{i}-\frac{\beta}{1-n\beta}F, \hspace{3em} i=1,...,n\,;\\
\ddot{y}+\sigma\dot{y}+\Omega^{2}y=\frac{\beta}{1-n\beta}F\,,
\end{cases}
\end{equation}
where 
\[
F:=-n(\sigma\dot{y}+\Omega^{2}y)+\sum_{i=1}^{n}\left(\theta_{i}+\dot{\theta}_{i}^{2}\theta_{i}-F_{i}\right)\,.
\]
This transformation is needed to reduce the system to a first order form.
In system \eqref{SecDerSys}
\[
\mu:=\frac{\beta}{1-n\beta}=\frac{\frac{m}{M+nm}}{1-\frac{nm}{M+nm}}=\frac{m}{M}
\]
is a convenient choice for the small parameter. Since there can
only be one small parameter we shall assume that $\varepsilon=\mu a$ and set $\F_{i}:=a(\g^{2}-\t_{i}^{2})\dot{\t}_{i}$,
so that $F_{i}=\mu\F_{i}$. The system transforms into 
\begin{equation}\label{FiSys}
\begin{cases}
\ddot{\theta}_{i}+\theta_{i}-\sigma\dot{y}-\Omega^{2}y=\mu\left(\F_{i}-F\right), \hspace{3em} i=1,...,n\,;\\
\ddot{y}+\sigma\dot{y}+\Omega^{2}y=\mu F\,,
\end{cases}
\end{equation}
and its Poincar\'e analysis for $n=2$ will be our main goal in the next section.

In this section however, we wish to consider a simplified case when
$\s$ is also considered small. We will often refer to $\ds{\s=\frac{c}{(M+nm)\sqrt{g/l}}}$ 
as "damping in the frame", although it is also obviously affected by lenghts and masses of the 
pendulums and mass of the frame. Aside from serving as a stepping stone to the general case, the case of small $\s$ 
compares directly to the treatment in Chapter 5 of \cite{Bl2}. Similar extra assumption
is made there, but system \eqref{DimssSys} is not transformed into \eqref{SecDerSys}, making it necessary to 
neglect the backreaction of the frame on the pendulums given by $\ddot{y}$. In line with
the above, we set $\s=\mu b$ and subsume the $\dot{y}$
term in \eqref{FiSys} into the right hand side. After neglecting
the terms of order $\mu^{2}$ we finally arrive at
\begin{equation}\label{SmallSigSys}
\begin{cases}
\ddot{\theta}_{i}+\theta_{i}-\Omega^{2}y=\mu\left(\F_{i}-\F\right), \hspace{3em} i=1,...,n\,;\\
\ddot{y}+\Omega^{2}y=\mu\F\,,
\end{cases}
\end{equation}
 with $\F_{i}$ as above and 
\[
\F:=-b\dot{y}-2\,\Omega^{2}y+\sum_{i=1}^{n}\left(1+\dot{\theta}_{i}^{2}\right)\theta_{i}\,.
\]
System \eqref{SmallSigSys} can be rewritten in the first order form
as $\dot{\mathbf{x}}=\mathbf{A}\mathbf{x}+\mu\mathbf{\Phi}(\mathbf{x})$, where $\mathbf{x}=(\t_{1},\dot{\t_{1}},\dots,\t_{n},\dot{\t_{n}},y,\dot{y})^{T}$.
From now on we only consider the case of two pendulums, $n=2$. Then
explicitly 
\[
\mathbf{A}=\text{\ensuremath{\left[\begin{array}{cccccc}
 0  &  1  &  0  &  0  &  0  &  0\\
-1  &  0  &  0  &  0  &  \Omega^{2}  &  0\\
0  &  0  &  0  &  1  &  0  &  0\\
0  &  0  &  -1  &  0  &  \Omega^{2}  &  0\\
0  &  0  &  0  &  0  &  0  &  1\\
0  &  0  &  0  &  0  &  -\Omega^{2}  &  0 
\end{array}\right]}\qquad and\qquad }\mathbf{\Phi}=\left[\begin{array}{c}
0\\
\F_{1}-\F\\
0\\
\F_{2}-\F\\
0\\
\F
\end{array}\right]\,.
\]
The corresponding linear generating system (see Appendix for the terminology) is obtained by setting
$\mu=0$ in \eqref{SmallSigSys}. Poincar\'e theorem gives sufficient
conditions for existence and stability of periodic solutions to the
original non-linear system that converge to periodic solutions of
the generating system when $\mu\to0$. The latter can be found most
easily by diagonalizing $\mathbf{A}=\mathbf{V}\mathbf{\Lambda}\mathbf{V}^{-1}$. In our case, 
$\mathbf{\Lambda}=\text{diag}\,(i,i,-i,-i,i\O,-i\O)$
and $\mathbf{x}=\mathbf{V}\mathbf{\xi}$ with $\mathbf{\xi}=(\pfi_{1},\pfi_{2},\psi_{1},\psi_{2},u,v)^{T}$
given by 
\begin{equation}\label{ChangeVar}
\begin{cases}
\varphi_{k}=\frac{1}{2}\left(i\theta_{k}+\dot{\theta}_{k}\right)+\frac{1}{2}\frac{\Omega^{2}}{\Omega^{2}-1}\left(iy+\dot{y}\right)\\
\psi_{k}=\frac{1}{2}\left(-i\theta_{k}+\dot{\theta}_{k}\right)+\frac{1}{2}\frac{\Omega^{2}}{\Omega^{2}-1}\left(-iy+\dot{y}\right)\\
u=\frac{1}{2}\left(i\Omega y+\dot{y}\right),\,\, v=\frac{1}{2}\left(-i\Omega y+\dot{y}\right)\,.
\end{cases}
\end{equation}
Let us further assume that the system is non-resonant, i.e. the natural frequency of the frame 
$\O$ is not an integer or a half-integer. In fact, it is realistic to expect $|\O|\ll1$,
meaning that the clocks oscillate much faster than the frame. Making
the corresponding substitution in non-linear terms, namely 
\begin{gather}
\mathbf{f}(\mathbf{\xi})=\mathbf{V}^{-1}\Phi(\mathbf{V}\mathbf{\xi})\label{fVxi},
\end{gather}
we get the diagonalized system 
\begin{gather}
\dot{\mathbf{\xi}}=\mathbf{\Lambda}\mathbf{\xi}+\mu \mathbf{f}(\mathbf{\xi})\label{fVxi}.
\end{gather}
The general solution to the generating system is $\dot{\mathbf{\xi}}=\mathbf{\Lambda}\mathbf{\xi}$
is $\pfi_{k}=\alpha_{k}\mathrm{e}^{it}$, $\psi_{k}=\beta_{k}\mathrm{e}^{-it}$, $u=C\mathrm{e}^{i\O t}$,
$v=D\mathrm{e}^{-i\O t}$. It is not periodic since $\O$ may be irrational,
and, in any case, the periodic solutions we are interested in should
reduce to independent oscillations of the pendulums when the coupling
is removed, i.e. when $\mu=0$. To put it differently, we are looking for
periodic solutions with the amplitude of frame motion of order at
most $\mu$. This means that we are looking for non-linear counterparts
of the solutions with $C=D=0$. 

Physically, parameter $\mu$ measures the strength of the non-linear coupling of the pendulums to the frame. The generating system describes two pendulums attached to an immovable frame, so they are decoupled. Once the frame is allowed to move the coupling is turned on and the behavior changes. It turns out that the pendulums eventually settle into oscillating synchronously with limit amplitudes determined by their escapement mechanisms. Depending on the initial phases this settled motion can be in-phase or anti-phase. In the anti-phase regime the frame remains almost motionless once the synchronization occured, the pendulums behave as if they were nearly decoupled. More sizable motions are required to sustain the anti-phase regime, speeding up the oscillations compared to decoupled pendulums. The next theorem makes this description more precise providing explicit formulas for limit amplitudes and periods to the first order in $\mu$.
It also implies that both regimes reemerge after small disturbances. 
\begin{theorem}\label{Th3Dsmall}
For small $\mu>0$, $a,b>0$ and $2\O$ a non-integer system \eqref{SmallSigSys} with $n=2$ admits both
in-phase and anti-phase synchronized periodic solutions with limit amplitude
$2\g$, where $\g\neq0$ is the critical angle in the van der Pol escapement.
Aside from the trivial one, these are the only solutions that converge when $\mu\to0$ to pendulums independently oscillating with equal amplitudes, and motionless frame. The period of the anti-phase solution is 
$T(\mu)=2\pi+o(\mu)$, and that of the in-phase solution is $T(\mu)=2\pi-\frac{1+\gamma^2}{1-\O^2}\,\mu+o(\mu)$. 
Both solutions are stable. 
\end{theorem} 
\begin{proof}
In this proof we use notation and terminology associated with the Poincar\'e theorem, see Appendix.
Recall that the diagonalized matrix of the generating system is $\mathbf{\Lambda}=\text{diag}\,(i,i,-i,-i,i\O,-i\O)$. 
We want our periodic solutions to reduce to oscillations of decoupled pendulums for $\mu=0$ meaning 
that only $\alpha_1,\dots,\alpha_4$ in \eqref{alphas} should be non-zero. Therefore, we must choose 
the first four eigenvalues to form the leading special group. The remaining two then belong to 
non-special groups, and the secondary special group is empty by assumption about $\O$. 

Our next task is to simplify constraint \eqref{eq:main_conds} from the Poincar\'e theorem for our case. By definition, $\xi=(\pfi_{1},\pfi_{2},\psi_{1},\psi_{2},u,v)^{T}$ and the relevant components of $\mathbf{\xi}$ are the first four. Since 
$\mathbf{x}=\mathbf{V}\mathbf{\xi}$ must be a real solution one can see from \eqref{ChangeVar} that $\phi_{k}$ and $\psi_{k}$ must be complex conjugates of each other for $k=1,2$. In view of \eqref{alphas} we have 
$\alpha_{k+2}=\overline{\alpha_{k}}$, and we represent $\alpha_{k}=r_{k}\mathrm{e}^{i\phi_{k}}$. 
Here $r_{k}$ is the half-amplitude of the $k$th pendulum and $\phi_{k}$ is its
phase. One can set $\phi_{2}=0$ without loss of generality since
time can be shifted by a constant, then synchronization is characterized
by a single phase $\phi_{1}=\phi$. By assumption about equal amplitudes 
set $r_{1}=r_{2}=r$. After some computations one can reduce 
\eqref{eq:main_conds} to a system of equations for $r$ and $\phi$, namely 
\begin{alignat}{1}\label{smallsamps}
\begin{cases}
\left(\gamma^{2}-r^{2}\right)r^{2}\mathrm{e}^{i\phi}=0\\
\gamma^{2}\left(1+\gamma^{2}\right)\sin(\phi)=0\\
r^{2}\left(1+r^{2}\right)\left(\mathrm{e}^{2i\phi}-1\right)=0\,.
\end{cases}
\end{alignat}
Discarding the trivial solution $r=0$, the only remaining ones are
$r=\g$, $\phi=0$ and $r=\g$, $\phi=\pi$, corresponding to the
in-phase and the anti-phase synchronization respectively. Substituting the values of $\alpha_{s}$ 
corresponding to each solution into \eqref{eq:per_cor} we get the first order period corrections.

Stability condition \eqref{eq:thm_eq1} reduces to the following equations having only roots with negative real parts
\begin{align*}
\phi=0:&\quad(\varkappa+a\g^{2})\Bigl((1-\O^{2})^{2}\varkappa^{2}+a\g^{2}(1-\O^{2})^{2}\varkappa+(1+\g^{2})^2\Bigr)=0\,;\\
\phi=\pi:&\quad(\varkappa+a\g^{2})\Bigl((1-\O^{2})^{2}\varkappa^{2}+a\g^{2}(1-\O^{2})^{2}\varkappa
+3\g^{4}+4\g^{2}+1\Bigr)=0\,.
\end{align*}
This condition is equivalent to all the coefficients of the above polynomials being positive, which
happens if $a>0$, $\g\neq0$ and $\O\neq\pm1$. Stability condition \eqref{eq:thm_eq3} for non-special eigenvalues reduces to $b>0$ in both cases. 
\end{proof}
Numerical simulations in Section \ref{S5} show that the limit amplitude usually gives a good approximation to actual amplitudes when $\mu$ is small. We believe that these are the only solutions to the Poincar\'e system \eqref{eq:main_conds} even without the assumption of equal amplitudes. However, we were unable to solve it for $r_{1}$, $r_{2}$ and $\phi$ without assuming that $r_{1}=r_{2}$. We can only show that if $\phi=0$ or $\phi=\pi$ is assumed in advance then $r_{1}=r_{2}$ follows from the equations. Based on simulations, we also believe that for small $\mu$ these are the only attractors, i.e. the trivial solution is spurious (no periodic solutions with amplitudes of order $\mu$), and no quasiperiodic or chaotic behavior occurs.

We see that for small $\s$ there holds a counter-intuitive symmetry: both synchronization regimes always coexist and have the same stable amplitudes. This conclusion coincides with Blekhman's in \cite{Bl2} even though he neglects the second derivative of the frame's position and applies the Poincar\'e theorem to system \eqref{DimssSys} directly. However, such coexistence of regimes was not observed in some experiments, see \cite{Ben}. We shall see in the next section that this symmetry is a deceptive consequence of $\s$ being small, increasing it sufficiently eliminates the in-phase regime altogether.

\section{Damping in the frame and the in-phase synchronization}\label{S3}

In this section we no longer assume that damping in the frame is small
and investigate how its magnitude influences synchronization. It is this assumption, it turns out, that puts the two synchronization regimes on an equal footing in Theorem \ref{Th3Dsmall}. Recall
that \eqref{SmallSigSys} was obtained from \eqref{FiSys} after assuming that $\s$ is small. 
Without this assumption, after neglecting the terms of order $\mu^{2}$, we get instead 
\begin{equation}\label{SigSys}
\begin{cases}
\ddot{\theta}_{i}+\theta_{i}-\sigma\dot{y}-\Omega^{2}y=\mu\left(\F_{i}-\F\right)\\
\ddot{y}+\sigma\dot{y}+\Omega^{2}y=\mu\F
\end{cases}
\end{equation}
with $\F_{i}:=a(\g^{2}-\t_{i}^{2})\dot{\t}_{i}$ and 
\[
\F:=-n(\s\dot{y}+\O^{2}y)+\sum_{i=1}^{n}(1+\dot{\t}_{i}^{2})\t_{i}\,.
\]
As in the previous section, we rewrite this as a first order system 
$\dot{x}=\mathbf{A}\mathbf{x}+\mu\mathbf{\Phi}(\mathbf{x})$,
where $\mathbf{x}=(\t_{1},\dot{\t_{1}},\dots,\t_{n},\dot{\t_{n}},y,\dot{y})^{T}$.
For two pendulums one has
\[
\mathbf{A}=\left[\begin{array}{cccccc}
0 & 1 & 0 & 0 & 0 & 0\\
-1 & 0 & 0 & 0 & \Omega^{2} & \sigma\\
0 & 0 & 0 & 1 & 0 & 0\\
0 & 0 & -1 & 0 & \Omega^{2} & \sigma\\
0 & 0 & 0 & 0 & 0 & 1\\
0 & 0 & 0 & 0 & -\Omega^{2} & -\sigma
\end{array}\right]\text{\qquad and \qquad}\mathbf{\Phi}=\left[\begin{array}{c}
0\\
\F_{1}-\F\\
0\\
\F_{2}-\F\\
0\\
\F
\end{array}\right]
\]
The diagonalized matrix is $\mathbf{\Lambda}=\text{diag}\,(i,i,-i,-i,-\frac{1}{2}\left(\sigma-\sqrt{\sigma^{2}-4\Omega^{2}}\right),-\frac{1}{2}\left(\sigma+\sqrt{\sigma^{2}-4\Omega^{2}}\right)\,)$,
we omit the diagonalizing transformation $V$. 

Behavior of the system is most conveniently described in terms of a modified damping parameter
$\widetilde{\s}:=\frac{\s}{a\left(\left(1-\Omega^{2}\right)^{2}+\sigma^{2}\right)}$. For small values of 
$\Omega$ and $\sigma$ it is close to the ratio of damping $\sigma$ to the escapement strength $a$. According to the theorem below the system undergoes two transitions as $\sigma$ and hence $\widetilde{\s}$ increase. For values smaller 
than $\widetilde{\s}:=\frac{\gamma^2}{2(2+\gamma^2)}$ both synchronization regimes coexist and are stable. This agrees with the conclusions of the previous section. The settled frequency of the in-phase oscillations decreases as the damping in the frame grows, but still remains higher than for the anti-phase oscillations.

Once the above value is exceeded the in-phase regime destabilizes, but remains in the picture until the value 
$\widetilde{\s}:=\frac{\gamma^2}{2}$ is attained. At this point the in-phase regime disappears altogether. Parameter values in the experiments of Bennett et al. \cite{Ben} (and presumably Huygens) seem to fall within this last range. This means that if clocks are started with nearly equal phases they will eventually settle into anti-phase oscillations with amplitudes determined by the escapement mechanism. This agrees with observations and indicates that $\sigma$ should not be treated as a small parameter when the Poincare method is applied.
\begin{theorem}\label{Th3D} 
For small $\mu>0$ and $a,\s>0$ system \eqref{SigSys} with $n=2$ admits anti-phase
synchronized periodic solutions with the limit amplitude $2\g$, where $\g\neq0$
is the critical angle in the van der Pol escapement. In-phase synchronized
periodic solutions can only exist if $\widetilde{\s}<\frac{\gamma^2}{2}$, where
$\widetilde{\s}:=\frac{\s}{a}\left(\left(1-\Omega^{2}\right)^{2}+\sigma^{2}\right)^{-1}$,
and have amplitude $2\sqrt{\frac{\gamma^{2}-2\widetilde{\s}}{1+2\widetilde{\s}}}$. 
Aside from the trivial one, these are the only solutions that converge when $\mu\to0$ to pendulums independently 
oscillating with equal amplitudes, and motionless frame. In-phase solutions exist and are
stable at least if the following inequalities hold: either $a\sigma<1$ and 
$\widetilde{\s}<\frac{\gamma^2}{2(2+\gamma^2)}$, or $a\sigma>1$ and $\frac{a\s-1}{a\sigma}\frac{\gamma^2}{2}<\widetilde{\s}<\frac{\gamma^2}{2(2+\gamma^2)}$. The period of the anti-phase solutions is $T(\mu)=2\pi+o(\mu)$, and that of the in-phase solutions is $$T(\mu)=2\pi-\frac{(1-\O^2)(1+\gamma^2)}{(1-\O^2)^2+2\,\frac{\sigma}a+\sigma^2}\,\mu+o(\mu)\,.$$
\end{theorem} 
\begin{proof}
The general outline of the proof is the same as in Theorem \ref{Th3Dsmall}. The leading special group is again formed by the first four eigenvalues and the remaining two are non-critical for $\s>0$. We wish to simplify constraint 
\eqref{eq:main_conds} from the Poincar\'e theorem. After setting $\alpha_{k}=r_{k}\mathrm{e}^{i\phi_{k}}$ and taking into account that $\alpha_{k+2}=\overline{\alpha_{k}}$, $r_{1}=r_{2}=r$, $\phi_{2}=0$ and $\phi_{1}=\phi$ one can obtain: 
\begin{alignat*}{1}
\begin{cases}
\mathrm{e}^{i\phi}\left(\gamma^{2}-r^{2}\right)-\widetilde{\s}\mathrm{e}^{i\phi}\left(1+r^{2}\right)
-\widetilde{\s}\left(1+r^{2}\right)=0\\
\left[i\left(\Omega^{2}-1\right)-\sigma\right]\left(1+r^{2}\right)\mathrm{e}^{-i\phi}-\left[i\left(\Omega^{2}-1\right)+\sigma\right]\left(1+r^{2}\right)\mathrm{e}^{i\phi}\\
\hspace{0.25in}+2\left[\left(1-\Omega^{2}\right)^{2}+\sigma^{2}\right]\left(\gamma^{2}-r^{2}\right)a-2\sigma\left(1+r^{2}\right)=0\\
\left(\sigma+i\left(1-\Omega^{2}\right)\right)\left(1-\mathrm{e}^{-2i\phi}\right)\left(1+r^{2}\right)=0\,.
\end{cases}
\end{alignat*}
Last equation immediately implies that $\phi=0$ or $\phi=\pi$. Substituting these values into the first equation we get the desired half-amplitudes $r$ for the in-phase and the anti-phase regimes. The second equation is then satisfied automatically. Expressions for periods are obtained by substituting the values of $\alpha_{s}$ 
corresponding to each solution into \eqref{eq:per_cor}. 

Stability condition \eqref{eq:thm_eq1} for $\phi=\pi$ reduces to the equation
\[
(\varkappa+a\g^{2})\Bigl(\varkappa^{2}+a\left((1+4\widetilde{\s})\g^{2}+2\widetilde{\s}\right)\varkappa
+a\frac{\widetilde{\s}}{\s}(1+\g^{2})\left(1+\g^{2}(a\s+3)\right)\Bigr)=0
\]
having only roots with negative real parts. For this to hold, all the coefficients must be positive. 
One can see by inspection that they are as long as $a,\s>0$ and $\g\neq0$. The same stability condition
for $\phi=0$ is somewhat more involved:
\[
\left(\varkappa+a(\g^{2}-2\widetilde{\s})\right)
\Bigl(\varkappa^{2}+a\frac{\g^{2}-2\widetilde{\s}(2+\g^{2})}{1+2a\widetilde{\s}}\varkappa
+a\frac{\widetilde{\s}(1+\g^{2})}{\s(1+2\widetilde{\s})^2}\left(1+2a\s\widetilde{\s}+\g^{2}(1-a\s)\right)\Bigr)=0.\\
\]
If $\widetilde{\s}<\frac{\gamma^2}{2(2+\gamma^2)}$ and $a\sigma<1$, or 
$a\sigma>1$ and $\frac{a\s-1}{a\sigma}\frac{\gamma^2}{2}<\widetilde{\s}<\frac{\gamma^2}{2(2+\gamma^2)}$ then the coefficients in the second parentheses are positive. In both cases $\widetilde{\s}<\frac{\gamma^2}{2}$ is automatically satisfied, so that the coefficient in the first parentheses is also positive. Other stability conditions (\ref{eq:thm_eq2},\ref{eq:thm_eq3}) do not apply here since the remaining eigenvalues are non-critical. 
\end{proof} 
Note that the no resonance condition on $\O$ is no longer necessary because $\s>0$ suppresses any resonance. Moreover, 
setting $\O=0$, i.e. eliminating the stiffness altogether, does not qualitatively alter the conclusions.
In practice, when the escapement has an engagement threshold the in-phase regime disappears even sooner, 
at least when the asymptotic amplitude falls below the threshold. The in-phase regime is then precluded by one of the pendulums switching off completely. For this to happen, it is not even necessary for asymptotic amplitude 
to be below the threshold as long as the pendulum angles fall below it before the onset of synchronization. Numerical simulations (see Section \ref{S5}) also show that when the in-phase regime is unstable, a period of beats may precede the anti-phase synchronization. In the course of beating the two pendulums exchange much of their energy back and forth with pendulum angles falling far below their initial or eventual asymptotic values. This is likely the reason why vanishing of oscillations in one of the clocks was observed in the experiments of Bennett et al. \cite{Ben}.

\section{Numerical simulations}\label{S5}

In this section we use numerical simulations to illustrate the existence and stability conditions of Theorem
\ref{Th3D}. 
\begin{figure}[htbp]
\begin{centering}
(a)\includegraphics[scale=0.33]{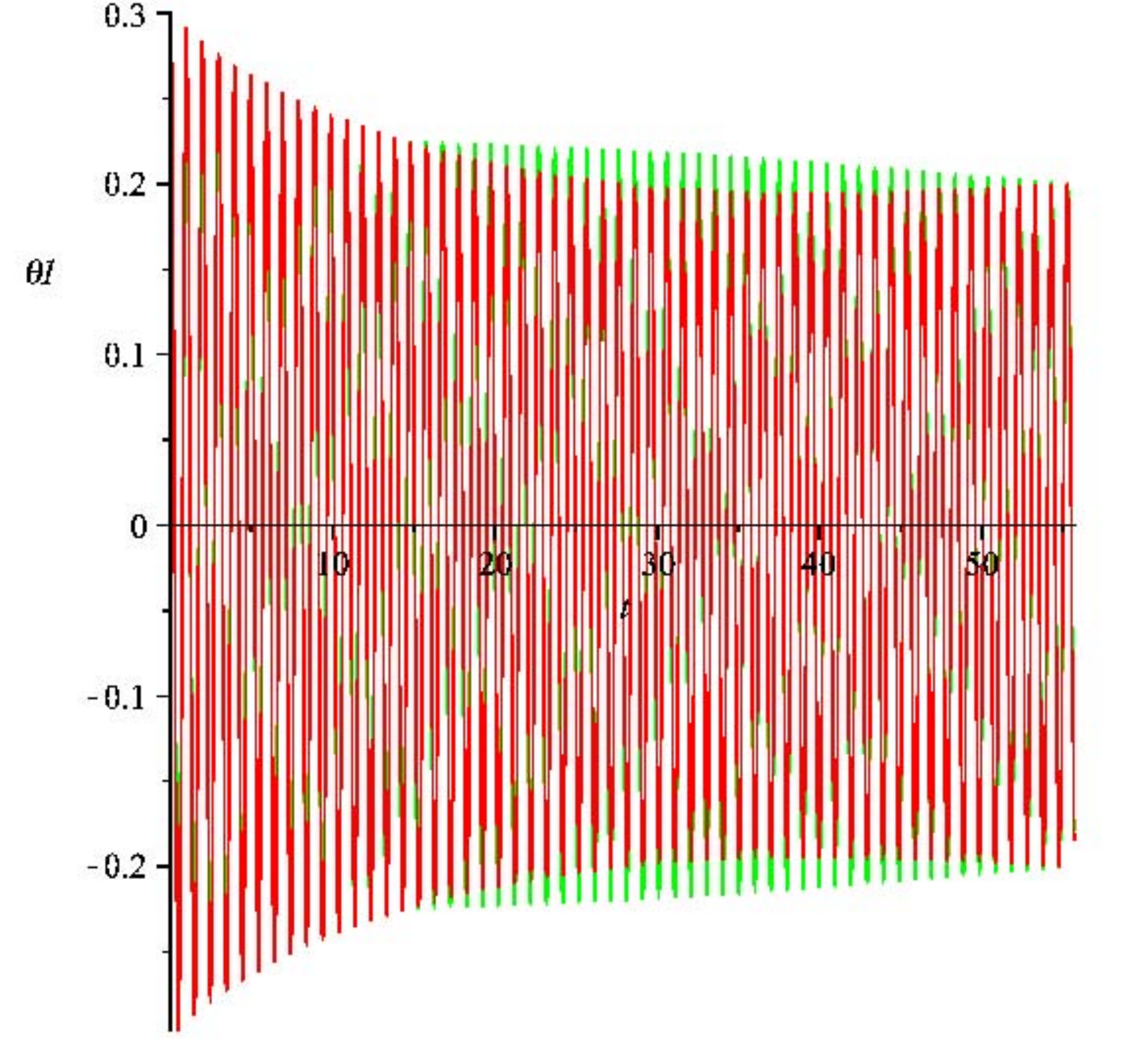}\quad(b)\includegraphics[scale=0.33]{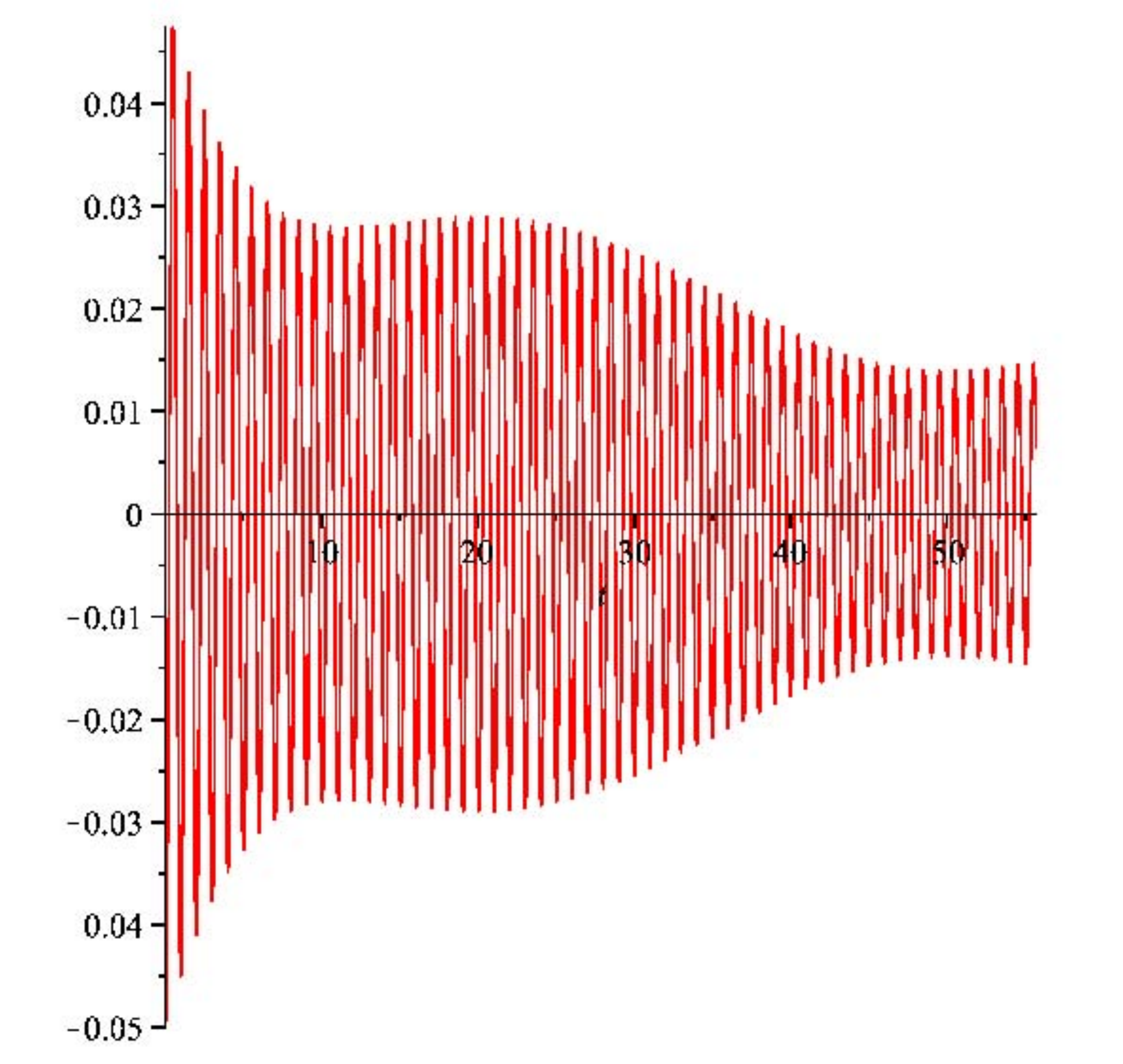}
\caption{\label{Czin}In-phase synchronization for parameters of Czolczynski et al. with $\theta_1(0)=0.25$ and $\theta_2(0)=0.3$. (a) $\theta_1(t)$ (lighter) and $\theta_2(t)$ (darker) superimposed, (b) $\theta_1(t)-\theta_2(t)$.}
\end{centering}
\end{figure}
We also wish to demonstrate that the theory based on the Poincar\'e method can explain, at least qualitatively, conflicting reports on the type of synchronization observed in experiments. 
Correspondingly, we tried to use the values of parameters from the experiments of Bennett et al.
\cite{Ben} and Czolczynski et al. \cite{Cz2}. Unfortunately, in the case of the former the information on these values is insufficient. 
\begin{figure}[htbp]
\begin{centering}
(a)\includegraphics[scale=0.32]{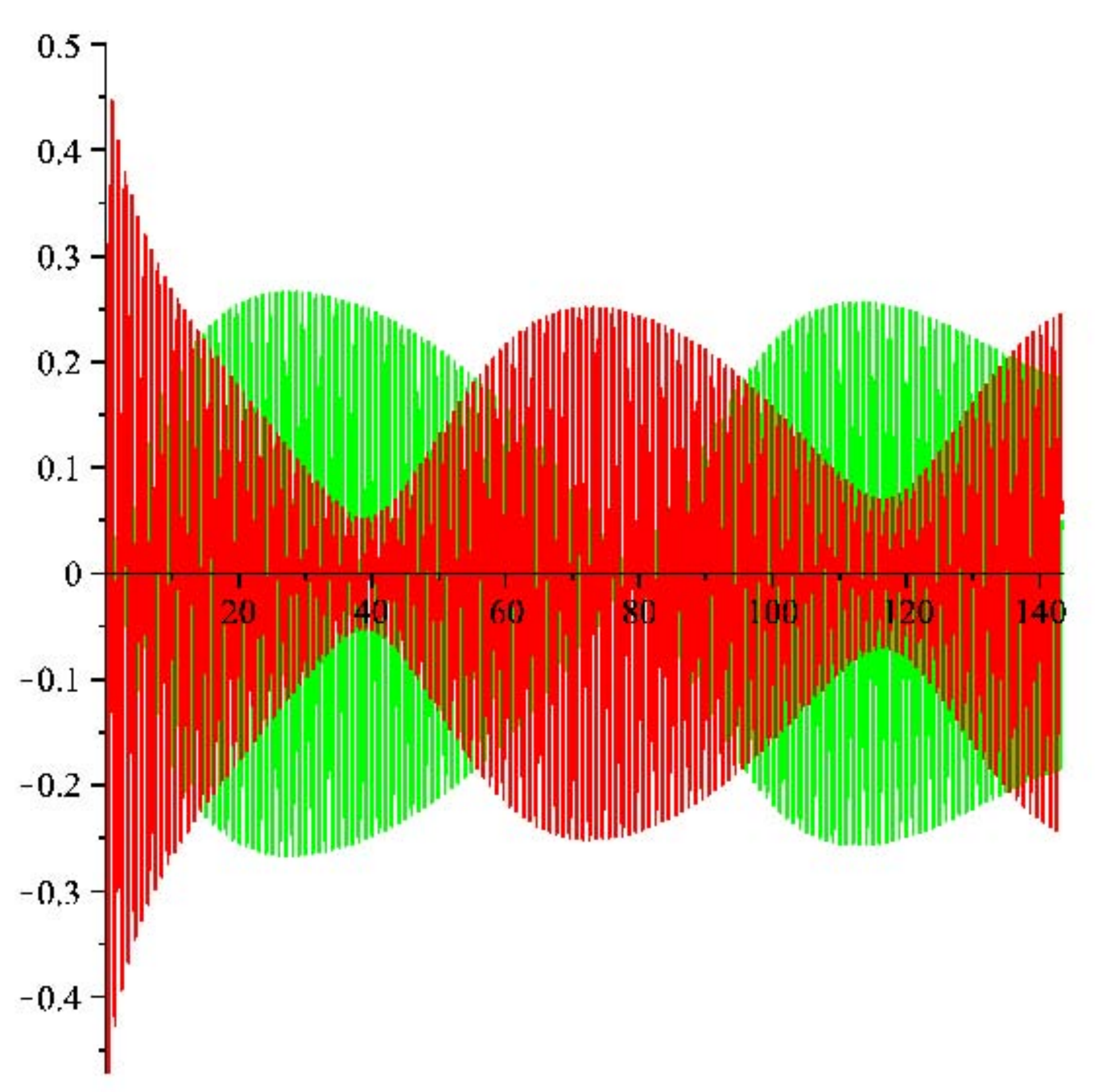}\quad(b)\includegraphics[scale=0.32]{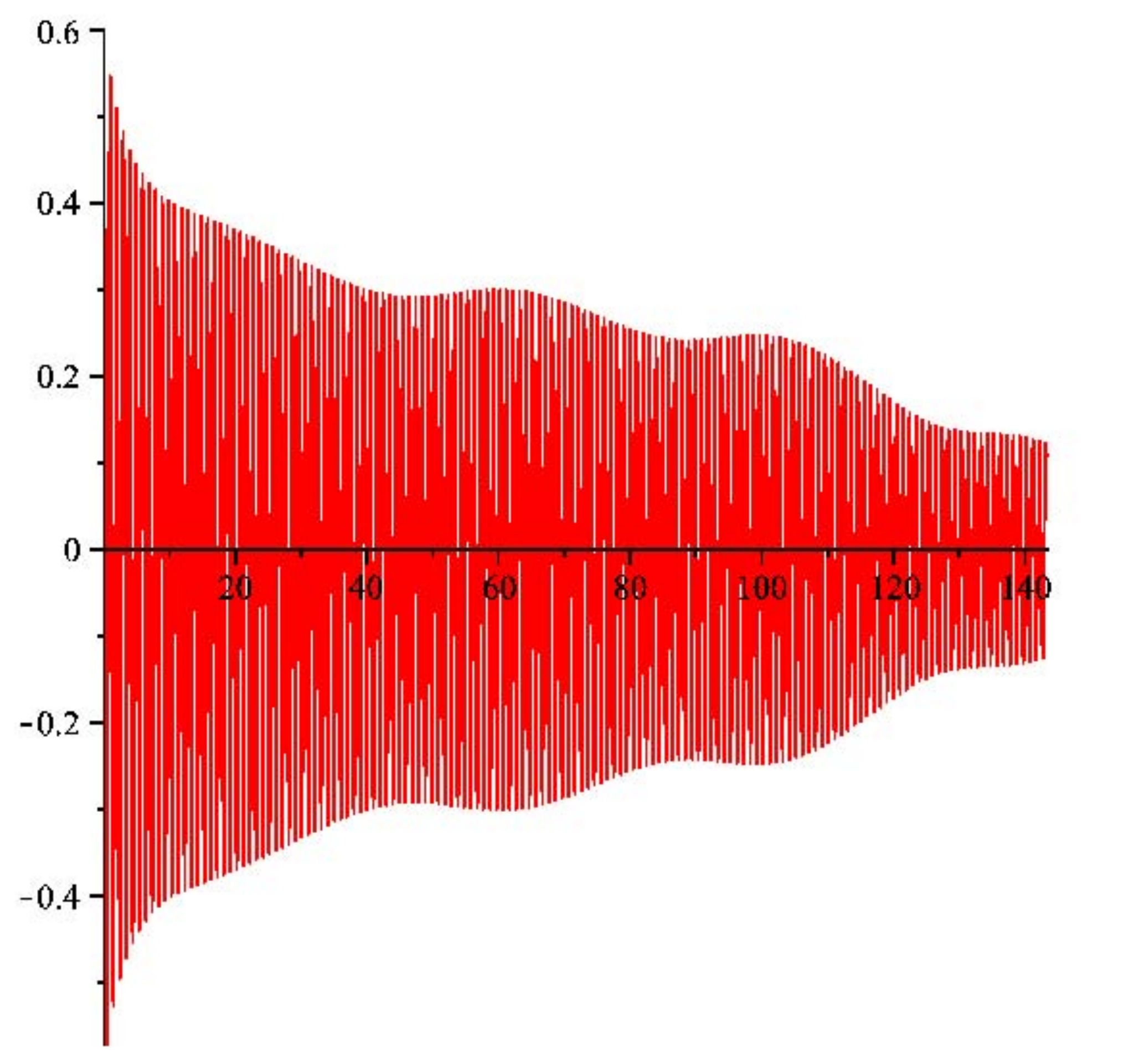}
\caption{\label{Czbeats}Anti-phase synchronization with beats for parameters of Czolczynski et al. with $\theta_1(0)=0.1$ and $\theta_2(0)=0.5$. (a) $\theta_1(t)$ (lighter) and $\theta_2(t)$ (darker) superimposed, (b) $\theta_1(t)+\theta_2(t)$.}
\end{centering}
\end{figure}

First three pairs of graphs correspond to the values $g=9.81$, $m=0.158$, $c=M=11.856$, $k=1.186$, $l=0.269$ from \cite{Cz2}, we picked $\g=0.122$ for the critical angle and $\varepsilon=5.047$ for the escapement strength. Time is measured in the number of cycles of a decoupled harmonic oscillator with the same parameters. One easily checks that $a\sigma<1$ and $\widetilde{\s}<\frac{\gamma^2}{2(2+\gamma^2)}$ in the notation from Theorem \ref{Th3D}, hence both synchronization regimes exist and are locally stable. When the initial values are close (Figure \ref{Czin}) the system synchronizes in-phase. 
\begin{figure}[htbp]
\begin{centering}
(a)\includegraphics[scale=0.3]{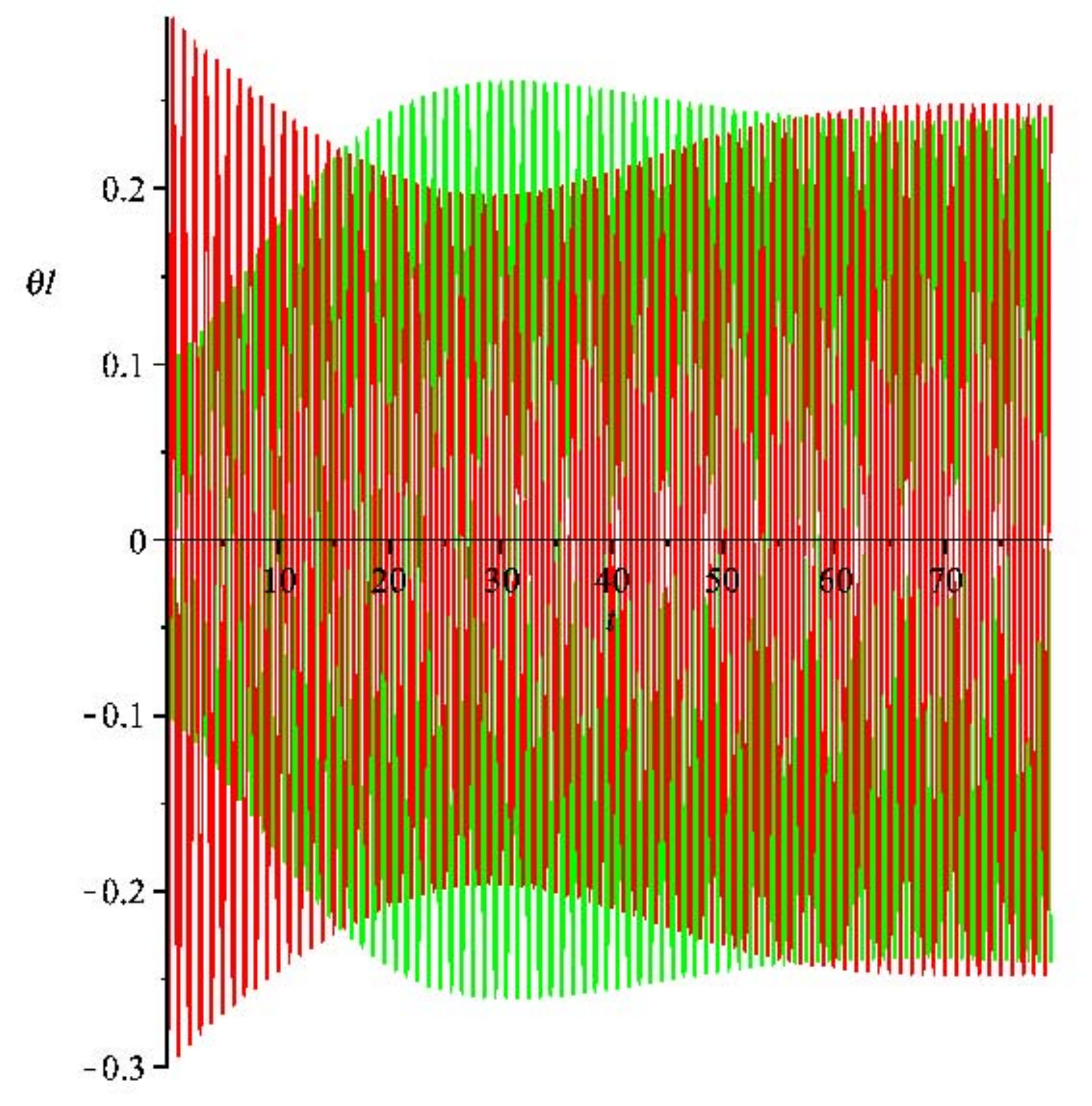}\quad(b)\includegraphics[scale=0.3]{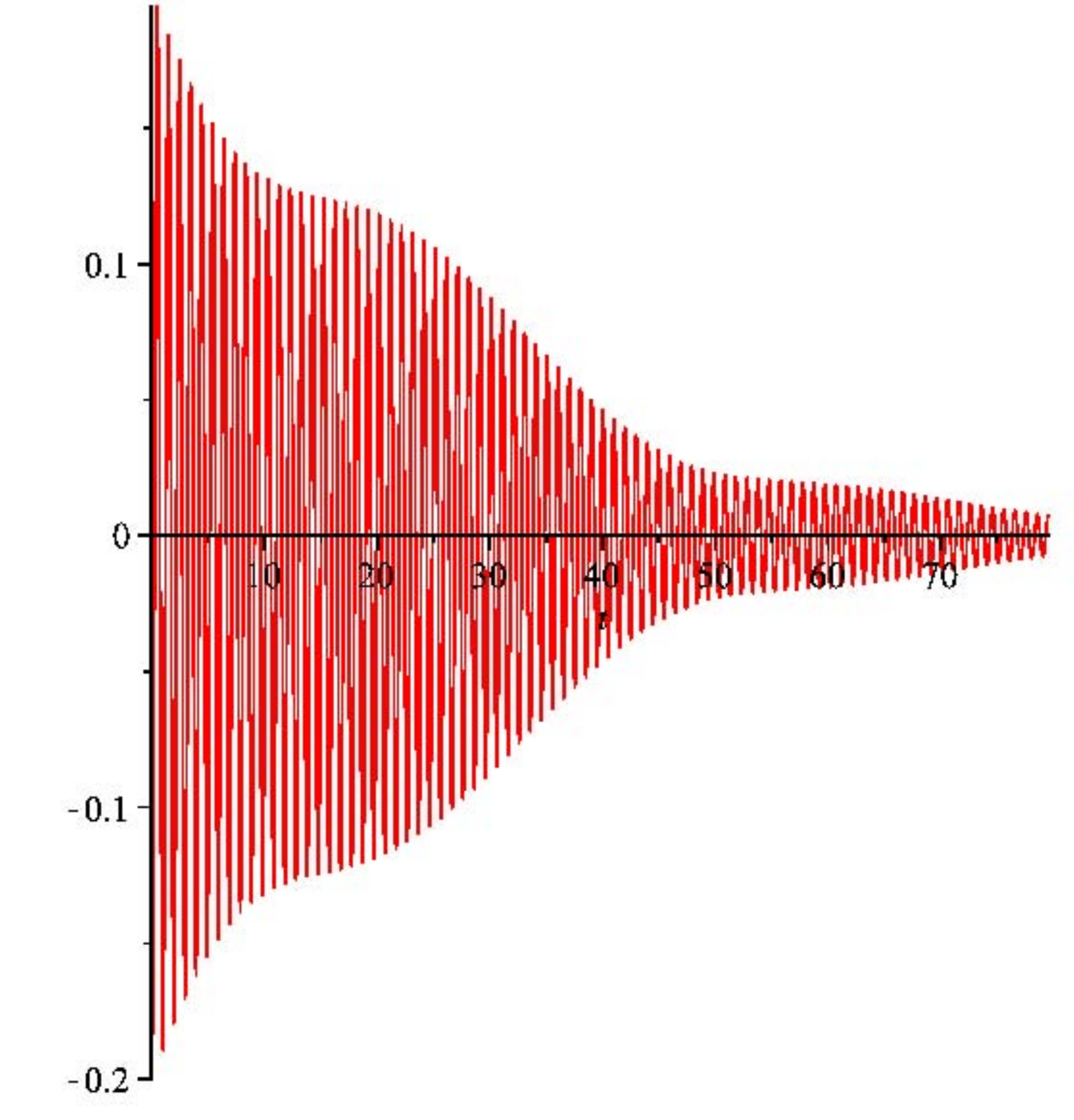}
\caption{\label{Czanti}Anti-phase synchronization for parameters of Czolczynski et al. with $\theta_1(0)=0.1$ and 
$\theta_2(0)=-0.3$. (a) $\theta_1(t)$ (lighter) and $\theta_2(t)$ (darker) superimposed, (b) 
$\theta_1(t)+\theta_2(t)$.}
\end{centering}
\end{figure}
As they are pulled apart, the asymptotic regime switches to anti-phase, but it is preceded by the depicted period of beats (Figure \ref{Czbeats}). Finally, when the initial values are sufficiently far apart (Figure \ref{Czanti}), the anti-phase synchronization is achieved faster and without beats. This is in line with the observations of Czolczynski et al. In all three cases the stable amplitudes are nearly equal to the limit values given by Theorem \ref{Th3D}.
\begin{figure}[htbp]
\begin{centering}
(a)\includegraphics[scale=0.3]{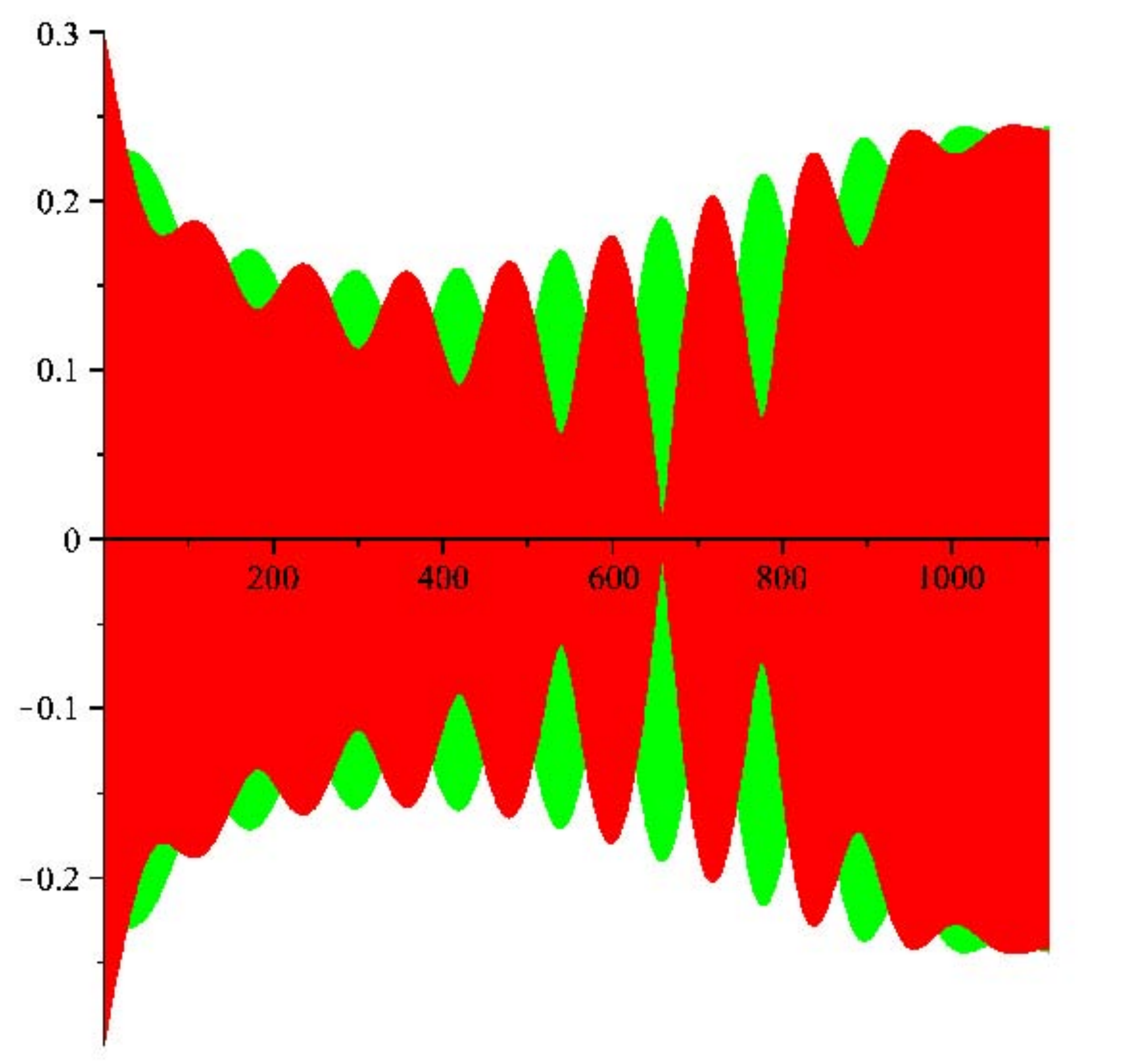}\quad(b)\includegraphics[scale=0.3]{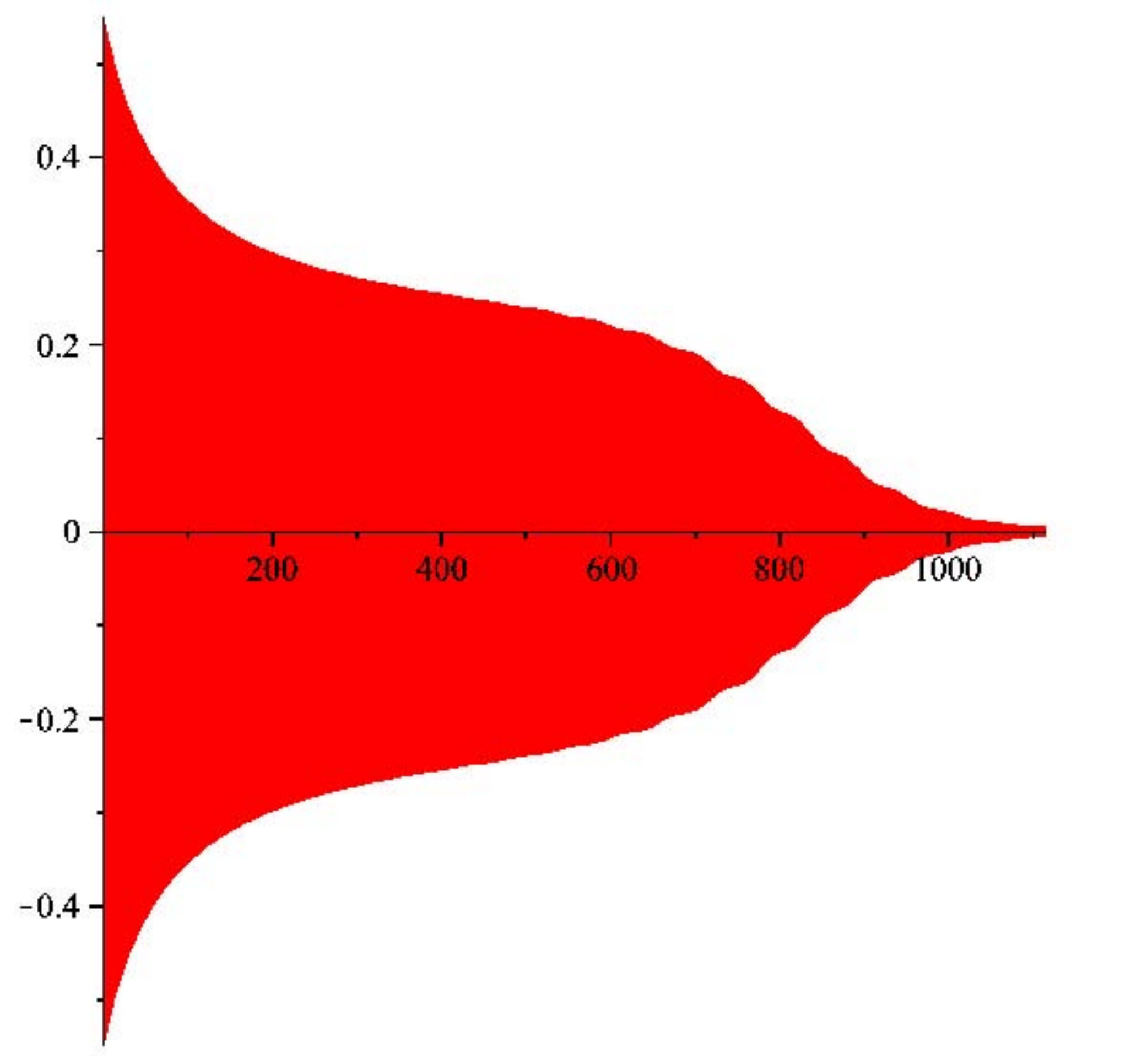}
\caption{\label{Bndeath}Anti-phase synchronization for parameters of Bennett et al., $M=9.716$, $c=k=M$, with $\theta_1(0)=0.25$ and $\theta_2(0)=0.3$. (a) $\theta_1(t)$ (lighter) and $\theta_2(t)$ (darker) superimposed, (b) 
$\theta_1(t)+\theta_2(t)$.}
\end{centering}
\end{figure}

Next, we tried to simulate the experiments of \cite{Ben} with $g=9.81$, $m=0.082$, $l=0.14$.  Unfortunately, there is no information on the values of stiffness $k$ or damping $c$. The mass $M$ was adjustable
in \cite{Ben} and we show results for two different choices of its value. We kept $\g=0.122$ and assigned the rest of the parameters as indicated in the captions. 
\begin{figure}[htbp]
\begin{centering}
(a)\includegraphics[scale=0.3]{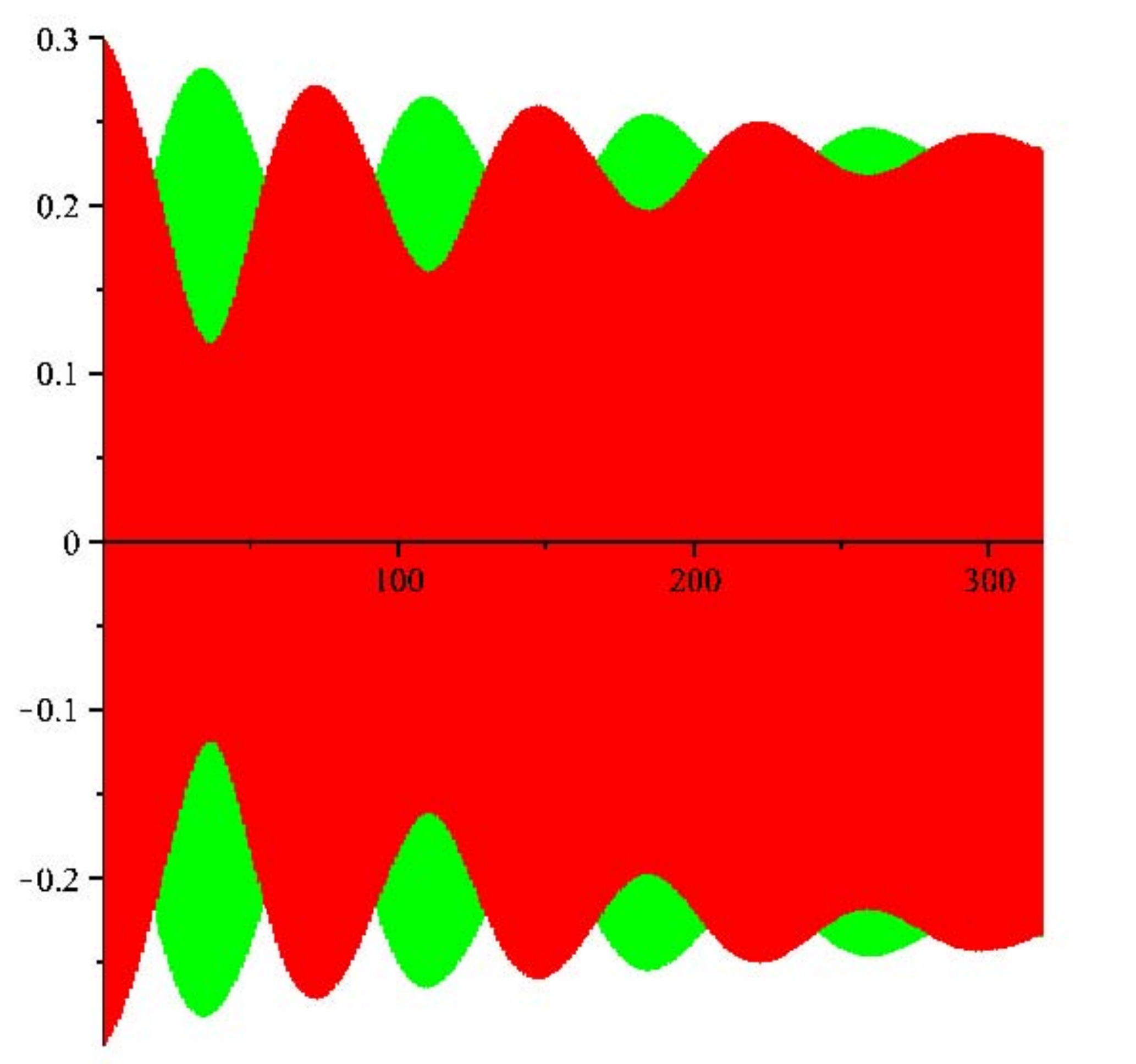}\quad(b)\includegraphics[scale=0.3]{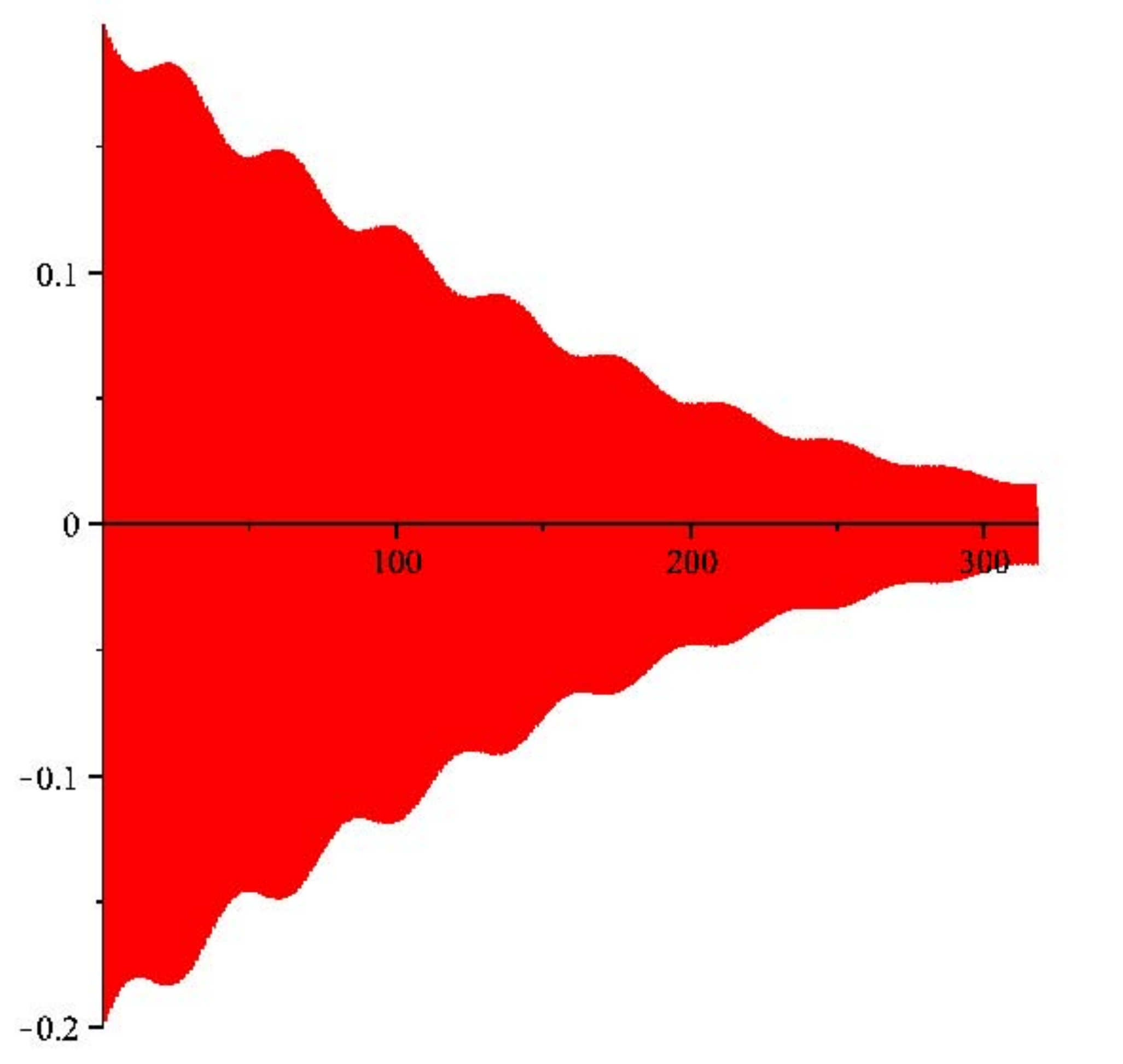}
\caption{\label{Bnbeats}In-phase synchronization for parameters of Bennett et al., $M=6.143$, $c=k=0.614$, with $\theta_1(0)=0.1$ and $\theta_2(0)=0.3$. (a) $\theta_1(t)$ (lighter) and $\theta_2(t)$ (darker) superimposed, (b) 
$\theta_1(t)-\theta_2(t)$.}
\end{centering}
\end{figure}
In the case of Figure \ref{Bndeath} the second inequality from Theorem \ref{Th3D} fails and the in-phase regime is unstable. The system synchronizes in anti-phase but only after rather violent beating (individual oscillations are not visible on the figure because of a large time interval). The onset of synchronization takes much longer than in the previous cases in line with the observations of Bennett et al. In fact, we adjusted the missing parameters to match the onset times from \cite{Ben}. Depending on the escapement threshold one may observe either the anti-phase synchronization or the vanishing of oscillations in one of the pendulums in this case. To stabilize the in-phase regime we had to decrease the damping by more than an order of magnitude (Figure \ref{Bnbeats}). We do not believe that such a 
value of $c$ is  realistic for the experiments of \cite{Ben} (or of Huygens) explaining why they never observed the in-phase synchronization. One can see that the onset of the in-phase synchronization can also be preceded by beats. 

As one can see, Theorem \ref{Th3D} provides some theoretical guidance to observing coexistence of synchronization regimes, or lack thereof, in experimental setups. First, the van der Pol parameters $a$ and $\gamma$ should be adjusted to match the system as closely as possible, e.g. along the lines of \cite{Pan}. Then the inequalities $a\sigma<1$ and 
$\widetilde{\s}<\frac{\gamma^2}{2(2+\gamma^2)}$ give the parameter range, where both types of synchronization can be expected to manifest (the other pair of inequalities from Theorem \ref{Th3D} calls for unrealistically large values of damping $\sigma$). When $\frac{\gamma^2}{2(2+\gamma^2)}<\widetilde{\s}<\frac{\gamma^2}{2}$ one can expect beats to precede the onset of anti-phase synchrionization due to the presence of unstable in-phase regime.

\section{Frame model with two masses}\label{S4}

Starting with Korteweg most work on the Huygens' problem was based on variations of the 
three degrees of freedom model we considered above. In it the casings of the clocks and the connecting body (beam, platform) are consolidated into a single massive frame. However, it is the casings that make the frame massive, the connecting body itself is actually quite light \cite{Ben,Cz2,Pan}. In this section we consider a refinement of this model suggested by Dil\~ao \cite{Dil}, where casings are modeled by two separate masses $M$, and the connecting body is modeled by a massles spring of stiffness $k$ and damping $c$, see Fig.\ref{fig:TwoMassSys}. The rigid frame of the previous model (with mass $2M$ and detached from a fixed support) is recovered when $k\to\infty$. 
\begin{figure}[htbp]
\centering
\includegraphics[width=3.2in]{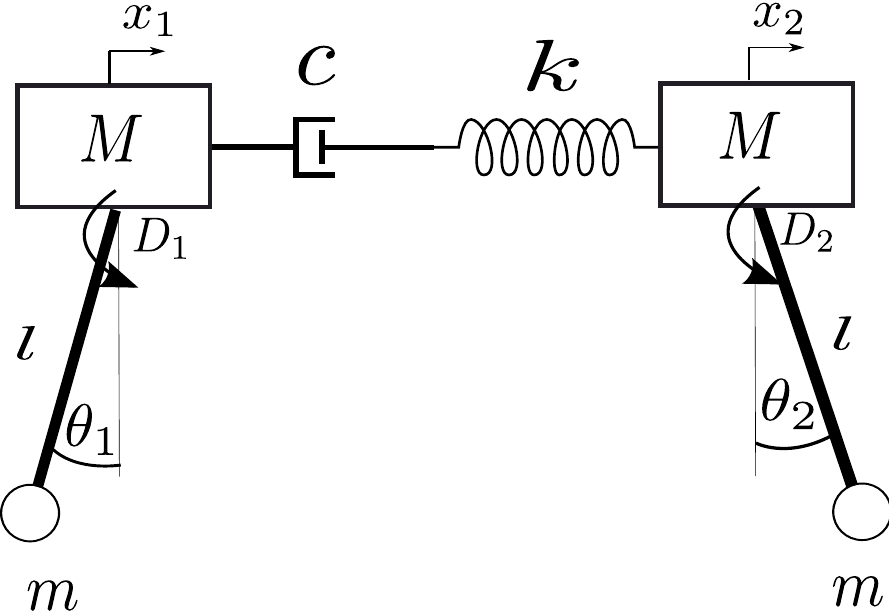}
\caption{\label{fig:TwoMassSys}Pendulums with clock casings connected by elastic frame.}
\end{figure}
Each mass carries its own pendulum of length $l$ and mass $m$, and, as before, the system is restricted to move horizontally. The positions of the casings are given by $x_{1}$, $x_{2}$ and the pendulum angles are 
$\theta_{1}$, $\theta_{2}$. Unlike Dil\~ao, who uses a piecewise-linear model for escapements, we continue to use the van der Pol terms $D(\t,\dot{\t}):=e(\g^{2}-\t^{2})\dot{\t}$ so that the Poincare method can be applied. 
Assuming small pendulum angles as before, the partially linearized equations of motion are
\begin{alignat*}{1}
\begin{cases}
ml^{2}\ddot{\theta}_{1}+mgl\,\theta_{1}+ml\,\ddot{x}_{1}=D_{1}\\
ml^{2}\ddot{\theta}_{2}+mgl\,\theta_{2}+ml\,\ddot{x}_{2}=D_{2}\\
(M+m)\ddot{x}_{1}+c\dot{x}_{1}-k\left(x_{2}-x_{1}\right)+ml\,\ddot{\theta}_{1}=ml\,\dot{\theta}_{1}^{2}\theta_{1}\\
(M+m)\ddot{x}_{2}+c\dot{x}_{2}+k\left(x_{2}-x_{1}\right)+ml\,\ddot{\theta}_{2}=ml\,\dot{\theta}_{2}^{2}\theta_{2}
\end{cases}
\end{alignat*}
with $D_{i}:=D(\t_{i},\dot{\t}_{i})$. Note, that Dil\~ao drops the non-linear terms on the right in the last two equations. The reason is the same as for his choice of a piecewise-linear escapement. Dil\~ao's approach requires evolution of the system to be linear as long as $\t_i$ stay below a critical angle, at which the escapment switches from boosting to damping. In contrast, the Poicare method only requires that the non-linear terms be analytic. Rescaling time as in Section \ref{S1}, and introducing dimensionless parameters $x_{1}=ly_{1}$ and $x_{2}=ly_{2}$ we obtain a dimensionless system 
(compare to \eqref{DimssSys}):
\begin{alignat*}{1}
\begin{cases}
\ddot{\theta}_{1}+\theta_{1}+\ddot{y}_{1}=F_{1}\\
\ddot{\theta}_{2}+\theta_{2}+\ddot{y}_{2}=F_{2}\\
\ddot{y}_{1}+\sigma\dot{y}_{1}-\varkappa\left(y_{2}-y_{1}\right)+\beta\ddot{\theta}_{1}=\dot{\theta}_{1}^{2}\theta_{1}\\
\ddot{y}_{2}+\sigma\dot{y}_{2}+\varkappa\left(y_{2}-y_{1}\right)+\beta\ddot{\theta}_{2}=\dot{\theta}_{2}^{2}\theta_{2}\,,
\end{cases}
\end{alignat*}
where $F_{i}=\frac{D_{i}}{mgl}=\frac{e}{mgl}(\g^{2}-\t_{i}^{2})\dot{\t}_{i}$, $\sigma=\frac{c}{\left(M+m\right)\sqrt{g/l}}$, $\varkappa=\frac{kl}{\left(M+m\right)g}$ and $\beta=\frac{m}{M+m}$. We choose our small parameter to be 
$\mu:=\frac{\beta}{1-\beta}=\frac{m}{M}$ and set $\frac{e}{mgl}=\mu a$ as before. Separating terms multiplied by $\mu$ and discarding higher order terms we transform the system into the form
\begin{equation}\label{2mass_sys}
\begin{cases}
\ddot{\theta}_{1}+\theta_{1}-\sigma\dot{y}_{1}+\varkappa\left(y_{2}-y_{1}\right)
=\mu\,\left(\,a(\g^2-\theta_{1}^2)\,\dot{\theta}_{1}-\F_{1}\right)\\
\ddot{\theta}_{2}+\theta_{2}-\sigma\dot{y}_{2}-\varkappa\left(y_{2}-y_{1}\right)
=\mu\,\left(\,a(\g^2-\theta_{2}^2)\,\dot{\theta}_{2}-\F_{2}\right)\\
\ddot{y}_{1}+\sigma\dot{y}_{1}-\varkappa\left(y_{2}-y_{1}\right)=\mu\F_{1}\\
\ddot{y}_{2}+\sigma\dot{y}_{2}+\varkappa\left(y_{2}-y_{1}\right)=\mu\F_{2}\,,
\end{cases}
\end{equation}
where 
\[
\F_{1}=-\sigma\dot{y}_{1}+\varkappa\left(y_{2}-y_{1}\right)+\theta_{1}\left(1+\dot{\theta}_{1}^{2}\right)\,\,\,
\text{and\,\,\,\,\ensuremath{\F_{2}}=\ensuremath{-\sigma\dot{y}_{2}-\varkappa\left(y_{2}-y_{1}\right)
+\theta_{2}\left(1+\dot{\theta}_{2}^{2}\right)}.}
\]
System \eqref{2mass_sys} can be rewritten in the first order form 
$\dot{\mathbf{\xi}}=\mathbf{A}\xi+\mu\mathbf{\Phi}(\mathbf{\xi})$, where\\
$\mathbf{\xi}=(\t_{1},\dot{\t_{1}},\t_{2},\dot{\t}_{2},y_{1},\dot{y}_{1},y_{2},\dot{y}_{2})^{T}$ and
\[
\mathbf{A}=\text{\ensuremath{\left[\begin{array}{cccccccc}
0 & 1 & 0 & 0 & 0 & 0 & 0 & 0\\
-1 & 0 & 0 & 0 & \varkappa & \sigma & -\varkappa & 0\\
0 & 0 & 0 & 1 & 0 & 0 & 0 & 0\\
0 & 0 & -1 & 0 & -\varkappa & 0 & \varkappa & \sigma\\
0 & 0 & 0 & 0 & 0 & 1 & 0 & 0\\
0 & 0 & 0 & 0 & -\varkappa & -\sigma & \varkappa & 0\\
0 & 0 & 0 & 0 & 0 & 0 & 0 & 1\\
0 & 0 & 0 & 0 & \varkappa & 0 & -\varkappa & -\sigma
\end{array}\right]}\qquad and \qquad}\mathbf{\Phi}=\left[\begin{array}{c}
0\\
a(\g^2-\theta_{1}^2)\,\dot{\theta}_{1}-\F_{1}\\
0\\
a(\g^2-\theta_{2}^2)\,\dot{\theta}_{2}-\F_{2}\\
0\\
\F_{1}\\
0\\
\F_{2}
\end{array}\right]
\]
The diagonalized matrix is $\mathbf{\Lambda}=\text{diag}\left(i,i,-i,-i,-\frac{1}{2}\left(\sigma-\sqrt{\sigma^{2}-8\varkappa}\right),-\frac{1}{2}\left(\sigma+\sqrt{\sigma^{2}-8\varkappa}\right),-\sigma,0\right)$, we omit the diagonalizing transformation. Note that one of the eigenvalues is zero, which indicates the presence of a rigid body mode. Indeed, one can see directly from Fig.\ref{fig:TwoMassSys} that the system as a whole is free to move rigidly. In the original model this was prevented by attaching the frame to a fixed support. But since this attachment, represented by $\O^2$ in \eqref{SigSys}, is a minor influence on synchronization we dropped it here for simplicity. As a compensation, we will only consider solutions that do not excite the rigid body mode. 

As the next theorem shows, qualitatively the behavior does not differ much from the case with the rigid frame. The in-phase regime eventually disappears, albeit at somewhat larger values of $\s$. The effect of the frame stiffness 
$\varkappa$ is similar in nature to that of the attachment stiffness $\Omega$ in Theorem \ref{Th3D}. One should also keep in mind that $M$ in the one mass frame model is twice its value in the two mass model here. This leads to an extra coefficient of $2$ in amplitude formulas and $\mu/2$ in place of $\mu$. With this in mind, the case of Theorem \ref{Th3D} for $\O=0$ is obtained in the limit $\varkappa\to\infty$, i.e. the limit of rigid frame.
\begin{theorem} For small $\mu>0$ and $a,\s>0$ system \eqref{2mass_sys} admits anti-phase synchronized periodic solutions only if $\gamma^{2}>\s_{an}$, where $\s_{an}:=\frac{\s}{a}\left(\left(2\varkappa-1\right)^{2}+\sigma^{2}\right)^{-1}$,  with amplitude $2\sqrt{\frac{\gamma^{2}-\s_{an}}{1+\s_{an}}}$ and period
$$
T(\mu)=2\pi-\frac{(2\varkappa-1)(1+\gamma^2)}{\left(2\varkappa-1\right)^{2}
+\frac{\sigma}a+\sigma^2}\,\frac{\mu}2+o(\mu)\,.
$$
It admits in-phase synchronized periodic solutions only if 
$\gamma^{2}>\s_{in}$, where $\s_{in}:=\frac{\s}{a}\left(1+\sigma^{2}\right)^{-1}$, with amplitude 
$2\sqrt{\frac{\gamma^{2}-\s_{in}}{1+\s_{in}}}$ and period 
$\ds{T(\mu)=2\pi-\frac{1+\gamma^2}{1+\frac{\sigma}a+\sigma^2}\,\frac{\mu}2+o(\mu)\,}$.
As before, $\g\neq0$ is the critical angle in the van der Pol escapement.
\end{theorem} 
\noindent The proof is similar to that of Theorems \ref{Th3Dsmall} and \ref{Th3D} and we omit it. The only structural difference is that the $0$ eigenvalue has to be included into the special leading group, but the corresponding equation from \eqref{eq:main_conds} is vacuously satisfied. Note the correspondence between $\s_{in}$ and $\widetilde{\s}$ from Theorem \ref{Th3D}.  Note also that $\s_{an}<\s_{in}$ for large $\varkappa$, meaning that the range of existence of the anti-phase regime is wider than of the in-phase one, with disparity increasing as $\varkappa\to\infty$. Unfortunatley, we were unable to simplify the stability conditions to a tractable form in this case. But, by a continuity argument from Theorem \ref{Th3D}, at least for large $\varkappa$ the in-phase regime should destabilize at 
smaller values of $\s$ than the anti-phase regime.

\section{Conclusions}

We used the Poincar\'e method to study models of two connected clocks synchronizing their oscillations, a phenomenon originally observed by Huygens. The oscillation angles are assumed to be small so that the pendulums are modeled by harmonic oscillators, and the mass ratio $\mu$ of the pendulum bobs to their casings is taken as the small parameter of the Poincar\'e method. Clock escapements are modeled by the van der Pol terms. The results are most complete for the three degrees of freedom model, where the connecting frame and the clock casings are consolidated into a single mass. The partial results on the four degrees of freedom model with casings represented by separate masses do not alter 
the qualitative picture.

The anti-phase synchronization regime is dominant in the sense that it is universally stable under variation of the system parameters. In this regime (to the first order in $\mu$) the frame is motionless, and the stable amplitudes and periods of the pendulums are the same as for the decoupled van der Pol oscillators. The in-phase synchronization regime may exist and be stable, exist and be unstable, or not exist at all depending on parameter values. As the damping in the frame, or more precisely the dimensionless parameter $\s=\frac{c}{(M+2m)\sqrt{g/l}}$, is increased the in-phase stable amplitude and period are decreasing until this regime first destabilizes and then disappears. Analytic conditions for existence and stability of synchronization regimes, and analytic expressions for their stable amplitudes and period corrections are derived based on the Poincar\'e theorem. In particular, they provide explicit relations between escapement's strength and critical angle, and damping in the frame, that distinguish experimental setups supporting and not supporting the in-phase synchronization.

Our results vindicate the intuitive explanation of Huygens' sympathy due to Korteweg, that the in-phase linear mode damps out faster than the anti-phase mode, because it requires more sizable motions of the frame, unless the damping is kept small on purpose. Interaction between the two modes explains the beats observed by Ellicott and others \cite{Kort}. Beats are most likely to be observed when the in-phase regime is present but unstable.

The following remarks are based on numerical simulations and are therefore somewhat speculative. When both regimes coexist, even if the in-phase one is unstable, the system may undergo a long period of beats before settling asymptotically. This is most likely to happen when the initial values are near the boundary between the basins of attraction of the regimes. In the course of beating the pendulum angles may fall significantly below their initial and asymptotic values. If escapement models with engagement threshold are used instead of the van der Pol terms the beats may cause one of the clocks to stop in finite time. When the damping is particularly heavy the system goes through a linear stage of anti-phase oscillations with asymptotic amplitude dependent on the initial values. If the engagement threshold is absent however, the clocks eventually synchronize in anti-phase with stable amplitudes. Thus, despite the limitations of the van der Pol model for escapements, we can account for all types of behavior observed in experiments 
\cite{Ben,Cz2} and, to some extent, predict quantitatively when they occur.

\section*{Appendix: Poincar\'e method}

\setcounter{equation}{0} \global\long\global\long\def\theequation{A.\arabic{equation}}

For convenience of the reader we briefly describe a method of studying periodic solutions to sytems of quasi-linear equations that we use in this paper. More detailed accounts can be found in \cite{And}, Ch.IX and \cite{Bl2}, Ch. 10, proofs are given in \cite{Bl1}, Ch.V.

Consider a first order autonomous quasi-linear system 
$\dot{\mathbf{x}}=\mathbf{A}\mathbf{x}+\mu\mathbf{\Phi}(\mathbf{x})$ and assume that the $l\times l$ matrix
$\mathbf{A}$ can be diagonalized as $\mathbf{A}=\mathbf{V}\mathbf{\Lambda}\mathbf{V}^{-1}$ by a nonsingular linear transformation $\mathbf{V}$. Set $\mathbf{F}(\xi)=\mathbf{V}^{-1}\mathbf{\Phi}(\mathbf{V}\xi)$, then a change of variables 
$\mathbf{x}=\mathbf{V}\mathbf{y}$ transforms the system into the form 
$\dot{\mathbf{y}}=\mathbf{\Lambda}\mathbf{y}+\mu\mathbf{F}(\mathbf{y})$. The linear system 
$\dot{\mathbf{x}}=\mathbf{A}\mathbf{x}$, as well as $\dot{\mathbf{y}}=\mathbf{\Lambda}\mathbf{y}$, is called the {\em generating system}. We are interested in the periodic solutions to the original system that converge to periodic solutions to the generating system when $\mu\to0$. The latter exist only if some eigenvalues of $\mathbf{A}$, the diagonal elements of $\mathbf{\Lambda}$, are purely imaginary, and they are stable only if $\mathbf{A}$ has no eigenvalues with positive real parts. From now on we assume that all eigenvalues have non-positive real parts, and call purely imaginary ones {\em critical}. In synchronization problems one is typically interested in particular solutions to the generating system, which correspond to synchronous motion of decoupled oscillators. 

Let $\l=i\omega$ be a critical eigenvalue. Following Blekhman \cite{Bl1,Bl2}, we collect all eigenvalues of the form 
$\l_s=in_s\omega$ with integer $n_s$ into the {\em leading special group}, and assume, without loss of generality, that they occur for $s=1,\dots,k$. Then the most general periodic solution of the generating system with period 
$T=\frac{2\pi}{\omega}$ is given by 
\begin{equation}\label{alphas}
y_{s}^{\circ}=\begin{cases}
\alpha_{s}\mathrm{e}^{in_{s}\omega t}, & s=1,\dots,k\\
0, & s=k+1,\dots,l\,,\end{cases}
\end{equation}
where $\alpha_{s}$ are complex amplitudes. According to the usual method of small parameter, one now looks for solutions 
$\mathbf{y}=\mathbf{y}(t,\mu,\alpha,\beta)$, with $y_s=y_{s}^{\circ}+\beta_s$ for $t=0$, to the original system as multivariable power series in $\mu,\alpha,\beta$. If such solutions are periodic the period will be of the form $T(\mu)=T(1-\delta(\mu))$, where $\delta(\mu)=\delta_1\mu+o(\mu)$, since $T(0)=T$ by assumption. 

Existence of periodic solutions to the original system now reduces to whether one can find power series that also satisfy
\begin{equation}
\mathbf{y}(T(\mu),\mu,\alpha,\beta)-\mathbf{y}(0,\mu,\alpha,\beta)=0.\label{Period}
\end{equation}
Indeed, by autonomy this implies that $\mathbf{y}$ is $T(\mu)$ periodic. This is a system of $l$ analytic equations with $l+1$ unknowns $\mu,\beta_1,\dots,\beta_l$, so one extra condition on $\beta$'s can be imposed. The system is solvable only if certain consistency constraints are satisfied. It turns out that $\mu$ factors out of the first $k$ equations and by setting their free terms to $0$ one obtains consistency constraints on $\delta_1$ and the amplitudes $\alpha_s$. 
Eliminating $\delta_1$ we finally end up with $k-1$ equations \eqref{eq:main_conds} for $k$ amplitudes. Again, one extra condition can be imposed on $\alpha$'s, e.g. $\alpha_{k-1}=\alpha_k$. This freedom reflects the possibility of shifting time in solutions to autonomous sytems. Thus, although the linear system has periodic solutions for arbitrary values of 
$\alpha_s$, non-linear periodic solutions only exist for some special collections of their values. Once the values 
of $\alpha_s$ are found from \eqref{eq:main_conds} the first order period correction $\delta_1$ can also be determined, see \eqref{eq:per_cor}.

If one is interested, as we are, in real-valued systems, then $\mathbf{A}$ is real and its eigenvalues come in complex conjugate pairs. The leading special group will contain both elements of a pair or none, so $k=2m$ and without loss of generality,
$\l_{m+s}=\overline{\l_{s}}$ for $s=1,\dots,m$. Generating periodic solutions will also be real if and only if 
$\alpha_{m+s}=\overline{\alpha_{s}}$.

Stability of a periodic trajectory can be investigated by looking at the solution to the matrix equation 
$\dot{\mathbf{Z}}=(\mathbf{A}+\mu\nabla\mathbf{\Phi}(x))\mathbf{Z}\,$ with $\,\mathbf{Z}(0)=\mathbf{I}$, where $\nabla\mathbf{\Phi}(x)$ is the derivative of the non-linear term computed along the trajectory. This is the system of variations. Stability is determined by eigenvalues $\rho(\mu)$ of 
$\mathbf{Z}(T(\mu))$, e.g. it suffices that they were less than one by absolute value. For $\mu=0$ we obviously have $\mathbf{Z}(t)=\mathrm{e}^{t\mathbf{A}}$, 
so $\rho_s(0)=\mathrm{e}^{\l_sT}$ and one can look for these eigenvalues in the form 
$\rho_s(\mu)=\mathrm{e}^{\l_sT}(1+\varkappa\mu+o(\mu))$. For non-critical eigenvalues $|\rho_s(0)|<1$ and we automatically have 
$|\rho_s(\mu)|<1$ for small $\mu$ by continuity. For critical eigenvalues we will have $|\rho_s(\mu)|<1$ for small 
$\mu>0$ as long as $\text{Re}(\varkappa)<0$. 

For the purposes of stating the stability conditions it is convenient to distinguish further groups among the critical eigenvalues. Two eigenvalues belong to the same group if they are integer multiples of each other, or if they differ by $in\omega$ with $n$ an integer. The point is that stability conditions decouple into separate conditions for each group. Recall that $\l_s=in_s\omega$ with integer $n_s$ form the leading special group. The {\em secondary special group} is $\l_s=i(n_s+\frac12)\omega$, and the remaining 
{\em non-special groups} have eigenvalues of the form $\l_s=i\nu_{s}^{(r)}$, where $r$ enumerates non-special groups.
After heavy duty computations the conditions we described can be brought into a convenient form below, which is due to Blekhman \cite{Bl2}. 
\bigskip
\begin{theorem*}[\textbf{Poincar\'e}]
Let $\dot{\mathbf{y}}=\mathbf{\Lambda}\mathbf{y}+\mu\mathbf{F}(\mathbf{y})$ be an autonomous quasi-linear system with diagonal $\mathbf{\Lambda}$. Assume that all eigenvalues of $\Lambda$ have non-positive real parts, and there is at least one purely imaginary eigenvalue $\l=i\omega$. For convenience, suppose that all eigenvalues of the form $\l_s=in_s\omega$ with integer $n_s$ are listed as first $k$. Then periodic solutions to the system becoming at $\mu=0$ periodic solutions to the generating system $\dot{\mathbf{y}}=\mathbf{\Lambda}\mathbf{y}$ with period $T=\frac{2\pi}{\omega}$ can only exist for
such values of amplitudes $\alpha_{1},...,\,\alpha_{k-2},\,\alpha_{k-1}=\alpha_{k}$,
that satisfy equations 
\begin{equation}
Q_{s}\left(\alpha_{1},...,\,\alpha_{k}\right)=\alpha_{k}n_{k}P_{s}-\alpha_{s}n_{s}P_{k}=0\,,\,\,\,\,\, s=1,...,k-1\,,
\label{eq:main_conds}
\end{equation}
where 
\begin{equation}
P_{s}\left(\alpha_{1},...,\,\alpha_{k}\right):=\left\langle F_{s}\mathrm{e}^{-in_{s}\omega t}\right\rangle:=\frac{1}{T}\int_{0}^{T}F_{s}\left(\alpha_{1}\mathrm{e}^{in_{1}\omega t},...,\,\alpha_{k}\mathrm{e}^{in_{k}\omega t},0,...,0\right)\mathrm{e}^{-in_{s}\omega t}dt\,,
\,\,\,\,\,s=1,...,k\,.
\end{equation}
The first order period correction for these solutions is 
\begin{equation}\label{eq:per_cor}
\delta_1=\frac{P_{k}\left(\alpha_{1}^{*},...,\alpha_{k}^{*}\right)}{i\alpha_{k}^{*}n_{k}T}\,,
\end{equation}
where $\alpha_{s}^{*}$ are the values found from \eqref{eq:main_conds}.

These solutions will indeed exist and be stable (asymptotically orbitally stable) for small positive $\mu$ 
if all the roots $\varkappa$ of the following algebraic equations have negative real parts:
\begin{equation}
\left|\frac{\partial Q_{s}}{\partial\alpha_{j}}-\alpha_{k}n_{k}\delta_{sj}\varkappa\right|=0,\,\,\,\,\,\,\,
s,j=1,...,k-1\,;
\label{eq:thm_eq1}
\end{equation}
\begin{equation}
\left|\left\langle \frac{\partial F_{s}}{\partial y_{j}}\mathrm{e}^{i\left(n_{j}-n_{s}\right)\omega t}\right\rangle -\delta_{sj}\left(\frac{\left(n_{s}+\frac{1}{2}\right)P_{k}}{n_{s}\alpha_{k}}+\varkappa\right)\right|=0,\,\,\,\,\,\,\, 
s,j\text{ in the secondary special group;}
\label{eq:thm_eq2}
\end{equation}
\begin{alignat}{1}
\left|\left\langle \frac{\partial F_{s}}{\partial y_{j}}\mathrm{e}^{i\left(\nu_{j}^{(r)}-\nu_{s}^{(r)}\right)\omega t}\right\rangle -\delta_{sj}\left(\frac{\nu_{s}^{(r)}P_{k}}{n_{s}\omega\alpha_{k}}+\varkappa\right)\right| & =0,
\,\,\,\,\,\,\,s,j\text{ in the $r$th non-special group.}
\label{eq:thm_eq3}
\end{alignat}
If the real part of at least one root of these equations is positive then the corresponding solution is unstable.
\end{theorem*}

\bibliographystyle{unsrt}
\bibliography{huygensref}

\begin{thebibliography}{10}

\bibitem{Kort}
D.~Korteweg.
\newblock Huygens' sympathic clocks and related phenomena in connection with
  the principal and the compound oscillations presenting themselves when two
  pendulums are suspended to a mechanism with one degree of freedom.
\newblock {\em Proceedings of the Royal Academy of Amsterdam}, 8:436--455,
  1905.

\bibitem{Ben}
M.~Bennett, M.~Schatz, H.~Rockwood, and K.~Wiesenfeld.
\newblock Huygens's clocks.
\newblock {\em The Royal Society of London. Proceedings. Series A.
  Mathematical, Physical and Engineering Sciences}, 458(2019):563--579, 2002.

\bibitem{Sen}
M.~Senator.
\newblock Synchronization of two coupled escapement-driven pendulum clocks.
\newblock {\em Journal of Sound and Vibration}, 291:566--603, 2006.

\bibitem{Bl2}
I.~Blekhman.
\newblock {\em Synchronization in science and technology}.
\newblock ASME Press, New York, 1988.

\bibitem{Cz2}
K.~Czolczynski, P.~Perlikowski, A.~Stefanski, and T.~Kapitaniak.
\newblock Huygens' odd sympathy experiment revisited.
\newblock {\em International Journal of Bifurcation and Chaos}.
\newblock To appear, DOI No: 10.1142/S0218127411029628.

\bibitem{Dil}
R.~Dil{\~ao}.
\newblock On the problem of synchronization of identical dynamical systems: the
  {H}uygens's clocks.
\newblock {\em Springer Optimization and Its Applications}, 33:163--181, 2009.

\bibitem{Frad}
A.~Fradkov and B.~Andrievsky.
\newblock Synchronization and phase relations in the motion of two-pendulum
  system.
\newblock {\em International Journal of Non-Linear Mechanics}, 42(6):895--901,
  2007.

\bibitem{Pan}
J.~Pantaleone.
\newblock Synchronization of metronomes.
\newblock {\em American Journal of Physics}, 70(10):992--1000, 2002.

\bibitem{Cz1}
K.~Czolczynski, P.~Perlikowski, A.~Stefanski, and T.~Kapitaniak.
\newblock Clustering of {H}uygens' clocks.
\newblock {\em Progress of Theoretical Physics}, 122(4):1027--1033, 2009.

\bibitem{Oud}
W.~Oud, H.~Nijmeijer, and A.~Pogromsky.
\newblock A study of {H}uijgens' synchronization: experimental results.
\newblock {\em Lecture Notes in Control and Information Sciences},
  336:191--203, 2006.

\bibitem{Kuz}
N.~Kuznetsov, G.~Leonov, H.~Nijmeijer, and A~Pogromsky.
\newblock Synchronization of two metronomes.
\newblock {\em Proceedings of the 3rd IFAC Workshop on Periodic Control
  Systems}, 3, 2007.

\bibitem{And}
A.~Andronov, A.~Vitt, and S.~Khaikin.
\newblock {\em Theory of oscillators}.
\newblock Dover, New York, 1987.

\bibitem{Bl1}
I.~Blekhman.
\newblock {\em Synchronization of dynamical systems}.
\newblock Nauka, Moscow, 1971.
\newblock (Russian).

\end{thebibliography}

\pagebreak
\thispagestyle{empty}

\end{document}